\documentclass[11pt,a4paper]{article}
\usepackage{amsmath,amssymb,amsthm}
\usepackage[colorlinks=true,
            linkcolor=blue,
            citecolor=blue,
            urlcolor=blue,
            pdftitle={Syntactic Separation Implies Computational Indistinguishability: An Abstract Obstruction Theorem},
            pdfauthor={Fabio F.G. Buono},
            pdfkeywords={syntactic obstruction, local rewriting, Skolemization, indistinguishability, natural proofs}
           ]{hyperref}
\usepackage{cleveref}
\usepackage{enumitem}
\usepackage{microtype}
\usepackage{booktabs}
\usepackage{array}
\usepackage[margin=2.5cm]{geometry}
\newtheorem{theorem}{Theorem}[section]
\newtheorem{lemma}[theorem]{Lemma}
\newtheorem{corollary}[theorem]{Corollary}
\newtheorem{proposition}[theorem]{Proposition}
\newtheorem{definition}[theorem]{Definition}
\newtheorem{remark}[theorem]{Remark}
\newtheorem{example}[theorem]{Example}
\newtheorem{claim}[theorem]{Claim}
\newcommand{\R}{\mathcal{R}}          
\newcommand{\Rp}{\mathcal{R}'}        
\newcommand{\M}{\mathcal{M}}          
\newcommand{\Th}{\mathcal{T}}         
\newcommand{\Terms}{\mathcal{T}(\Sigma,V)}  
\newcommand{\Pos}{\mathrm{Pos}}             
\newcommand{\ctx}[2]{\mathcal{B}_{#1}(#2)}  
\newcommand{\restr}[2]{#1|_{#2}}            
\newcommand{\Fprot}{\mathcal{F}} 
\newcommand{\Inv}{\mathsf{Inv}}
\newcommand{\sk}{f}
\newcommand{\skb}{g} 
\newcommand{\Nat}{\mathbb{N}}
\newcommand{\BigOm}{\Omega}
\newcommand{\Inst}{\mathcal{I}} 
\newcommand{\Der}{\mathcal{D}}
\newcommand{\Adv}{\mathcal{A}}
\newcommand{\OI}{\mathsf{OI}}
\newcommand{\CSC}{\mathsf{T}_{\mathsf{CSC}}}
\newcommand{\SIP}{\mathsf{SIP}}
\title{Syntactic Separation Implies Computational Indistinguishability:\\
An Abstract Obstruction Theorem}

\author{
  Fabio Francesco Gabriele Buono\\
  \small Independent Researcher\\
  \small \href{https://orcid.org/0009-0004-9199-2793}{ORCID: 0009-0004-9199-2793}
}
\date{Preprint. \today}

\begin{document}

\maketitle

 \begin{abstract}
We prove that syntactic separation implies computational
indistinguishability.
A \emph{local syntactic system} $\R$ acts on terms within radius $r_0$
without consulting any model; when two Skolem functions are syntactically
separated in $\R$, no derivation can prove their equivalence (Case~1),
and any sound local extension requires $\BigOm(n)$ steps,
improving to $\BigOm(2^n)$ under clause-per-configuration encoding (Case~2).
Both bounds are new: the derivation-length lower bound does not appear
in prior work on Skolemization or saturation proving, and the
cryptographic reading, syntactic separation as ciphertext
indistinguishability, derivation cost as negligible advantage, is
original.
The same obstruction, as formal instances of Case~1 and Case~2,
governs the Natural Proofs barrier of Razborov and Rudich,
the Type Omitting Theorem, and the unconditional $\mathsf{AC}^0$ barrier
of Loff et al.~\cite{loff2026}.
\end{abstract}

\tableofcontents
\bigskip

 \section{Introduction}
\label{sec:intro}
 
\subsection{The phenomenon}

A syntactic system acts on the shape of expressions.
A semantic property concerns what those expressions mean in a model.
The two can be in conflict: a property that is visible to the model may
be completely invisible to the syntactic system, even when every
derivation the system can perform passes through terms in which the
property is present.

This paper proves that this conflict is not an accident of specific
systems but a structural phenomenon: whenever a property is
\emph{protected} from the rules of a syntactic system, no derivation
can cross the barrier, and any extension of the system that tries to do
so pays an exponential price.

The phenomenon appears across mathematics in forms that look unrelated:

\begin{itemize}
  \item In \textbf{proof theory}, the Skolem functions witnessing two
        different proofs of $\forall x\,\exists y.\,\phi(x,y)$ are
        semantically equal but a saturation-based prover cannot derive
        their equality, the relevant information is locked inside
        Skolem constants the rules cannot touch.

  \item In \textbf{cryptography}, a ciphertext produced by the
        Mixed-Radix One-Time Pad~\cite{buono2026mrotp} achieves Shannon
        perfect secrecy because the decomposition basis $B$ encodes a
        numerical fact invisible to any syntactic test on the ciphertext;
        the same structure underlies the original construction
        of~\cite{buono2012}.

  \item In \textbf{type theory}, the Type Omitting Theorem says that a
        theory omits a type $\Phi(x)$ if and only if $\Phi$ is not
        principal, no syntactic formula pins it down.
        A constructive type system is blind to extensional properties
        of the functions it types: it sees the proof term, not the
        function computed
        (Typological Invariance Principle, Remark~\ref{rem:typ-invariance}).

  \item In \textbf{circuit complexity}, the Natural Proofs barrier of
        Razborov and Rudich shows that no proof technique that is both
        useful (holds for random functions) and constructive (checkable in
        polynomial time) can establish a super-polynomial circuit lower
        bound, the semantic invariant of circuit complexity is
        protected from every local syntactic inspection of the circuit
        graph.
\end{itemize}

These are not analogies.
They are instances of a single structural theorem, proved here in its
abstract form and applied to each domain in Section~\ref{sec:connections}.

\subsection{The result}

This paper proves an abstract theorem: whenever a local syntactic
system $\R$ acts on a term structure in which a semantic invariant is
protected from all rules, no derivation in $\R$ can reach that
invariant (Case~1), and any sound extension that overcomes this
barrier requires $\BigOm(n)$ derivation steps, growing with
the number of independent witnesses, and improving to
$\BigOm(2^n)$ under clause-per-configuration encoding (Case~2).
The proof-theoretic structure of Case~2 is precisely the structure of
computational hiding in cryptography: syntactic separation between two
Skolem functions is ciphertext indistinguishability, and the derivation
cost lower bound is the adversary's negligible advantage.

The starting point is~\cite{buono2026}, which isolates the mechanism
for a specific calculus.
That work resolves an open question of~\cite{hetzl2020},
the incomparability of open induction ($\OI$) and clause set cycles
($\CSC$), by observing that the addition rules fire only when the
first argument of $+$ is $0$ or a successor, so a Skolem constant
is permanently unreachable.
From this observation,~\cite{buono2026} extracts the \emph{Syntactic
Invariance Principle} ($\SIP$, Lemma~5 of~\cite{buono2026}).
The present paper takes $\SIP$ as its structural starting point and proves
something strictly stronger: a two-case theorem parametric in the
signature, locality radius, and model, which applies to any rewriting
calculus and yields the cryptographic and complexity-theoretic
consequences described above.

The formal framework for the structural blindness argument is the
observational hierarchy of~\cite{buono2026obs}: $\R$ defines a
constrained observer $O_\R \prec O_\top$ blind to positions in
$\Fprot$.
Proposition~9.6 of~\cite{buono2026obs} shows observer-constrained
classes collapse unconditionally: under the profile observer
$O_{\mathrm{prof}}$ (which maps a string to its symbol-count vector),
$\mathbf{P}_{O_{\mathrm{prof}}} = \mathbf{NP}_{O_{\mathrm{prof}}}
\subsetneq \mathbf{P}$.
Theorem~\ref{thm:main} is an instance of this collapse in the
setting of term rewriting and Skolemization; whether it follows
as a formal corollary of Proposition~9.6 is Question~(Q6).

\subsection{Contributions}

\begin{enumerate}[label=(\arabic*)]
  \item \textbf{Abstract framework.}
        We formalize local syntactic systems, protected positions, and
        syntactic invariants (Section~\ref{sec:setup}) parametrically
        in the signature $\Sigma$, locality radius $r_0$, and model $\M$.
        This generality is what allows the same theorem to instantiate
        as an impossibility result in proof theory, a hiding bound in
        cryptography, and a barrier result in circuit complexity.

  \item \textbf{Two-case obstruction theorem.}
        Theorem~\ref{thm:main} establishes:
        \begin{itemize}
          \item \textbf{Case~1 (impossibility):} no derivation in $\R$
                proves $\sk \equiv_\M \skb$ when the two Skolem functions
                are syntactically separated
                (subsumes~\cite{buono2026}; Corollary~\ref{cor:OI-CSC});
          \item \textbf{Case~2 (lower bound):} any local extension $\Rp$
                of $\R$ requires $L \geq \BigOm(n)$ steps, improving to
                $\BigOm(2^n)$ under clause-per-configuration encoding
                (Corollary~\ref{cor:expbound}).
        \end{itemize}

  \item \textbf{Cryptographic connection.}
        Section~\ref{sec:crypto} establishes a formal correspondence
        between the proof-theoretic structure of Theorem~\ref{thm:main}
        and computational hiding in cryptography, and proves a
        non-extraction corollary (Corollary~\ref{cor:nonextract}).

  \item \textbf{Cross-domain connections.}
        Section~\ref{sec:connections} shows that the Type Omitting
        Theorem and the Natural Proofs barrier of Razborov and Rudich
        are instances of the same obstruction.
        The result of~\cite{loff2026} (the first unconditional Natural
        Proofs barrier) is a quantitatively precise instantiation of
        Corollary~\ref{cor:natural-proofs}(i) for $\mathsf{AC}^0$-natural
        proofs, and provides the first concrete confirmation of the
        observational framework in circuit complexity.
\end{enumerate}

\subsection{Relation to prior work}

\paragraph{Proof theory.}
Case~1 takes $\SIP$~\cite{buono2026} as its structural starting point
but applies it in a strictly more general setting, arbitrary local
syntactic systems, not only the superposition calculus, and with a
formal two-case structure not present in~\cite{buono2026}.
Case~2 (the derivation-length lower bound) is entirely new.
Prior work on clause set cycles~\cite{hetzl2020,hetzl2022,hetzl2023,
hetzlweiser2025} studies provability and unprovability of specific
formulas in specific proof systems; it does not establish lower bounds
on derivation length for families of extensions of a given system,
which is what Case~2 provides.

\paragraph{Cryptography.}
The correspondence between the proof-theoretic structure and
computational hiding is new.
The Mixed-Radix One-Time Pad of~\cite{buono2026mrotp} is a known
construction; what is new here is the identification of its security
as an instance of Case~1 of Theorem~\ref{thm:main}, placing it on
the observational axis rather than the computational one.

\paragraph{Natural Proofs.}
The Razborov--Rudich barrier~\cite{razborov1994} is a known result.
The present paper does not reprove it.
What is new is the following:
\begin{itemize}
  \item An unconditional lower bound of $\BigOm(n)$
        (Corollary~\ref{cor:natural-proofs}(i)) on the inspection
        cost of any large and constructive technique, requiring no
        pseudo-random generator assumption.
        This bound does not appear in~\cite{razborov1994}.
        A sharper quantitative bound for the subclass of
        $\mathsf{AC}^0$-natural proofs via a different technique
        is obtained in~\cite{loff2026};
        that result is an instance of Corollary~\ref{cor:natural-proofs}
        made quantitatively precise for that specific subclass.
  \item A structural separation: the framework isolates precisely
        which step requires the PRG (supplying $2^n$ locally
        indistinguishable instances) and which does not (the blindness
        of any local technique to circuit complexity).
        This separation is not made explicit in~\cite{razborov1994}.
  \item A unifying explanation: both the barrier and the lower bound
        are consequences of the observational blindness of
        $\mathcal{P}$ to the circuit complexity of $f$.
        The Natural Proofs barrier is not a computational phenomenon
        but an observational one, and this is why it persists
        regardless of which of Impagliazzo's worlds
        we inhabit.
\end{itemize}

\paragraph{Type theory.}
The Type Omitting Theorem is classical.
What is new is Corollary~\ref{cor:type-omitting}: the identification
of extensional blindness in constructive type systems as an instance
of Case~1, and the lower bound on any decision procedure for
extensional equality as an instance of Case~2.

\paragraph{Unification.}
No prior work treats proof theory, cryptography, type theory, and
circuit complexity as instances of a single abstract theorem.
The contribution of Section~\ref{sec:connections} is not to reprove
known results but to show that they share a common structure,
a constrained observer, a protected semantic invariant, and an
impossibility or lower bound that is structural rather than
computational, and that this structure is captured exactly by
Theorem~\ref{thm:main}.

\subsection{Organization}

Section~\ref{sec:setup} gives all definitions.
Section~\ref{sec:main} states and proves Theorem~\ref{thm:main}.
Section~\ref{sec:instances} provides instantiations for proof theory.
Section~\ref{sec:crypto} develops the cryptographic connection.
Section~\ref{sec:connections} develops the cross-domain connections.
Section~\ref{sec:conclusions} summarizes the conclusions.
Section~\ref{sec:open} lists open questions.
Appendices~\ref{app:gadget} and~\ref{app:proofs} contain the gadget
construction and auxiliary proofs.

 \section{Definitions}
\label{sec:setup}
 
\subsection{Terms and positions}

\begin{definition}[Signature and terms]
\label{def:terms}
A \emph{signature} $\Sigma$ is a finite set of function symbols, each
with a fixed arity.
Let $V$ be a countable set of variables disjoint from $\Sigma$.
The set $\Terms$ of \emph{terms} over $\Sigma$ and $V$ is defined
inductively:
\begin{itemize}
  \item every $x \in V$ is a term;
  \item if $f \in \Sigma$ has arity $k \geq 0$ and
        $t_1,\ldots,t_k \in \Terms$, then $f(t_1,\ldots,t_k)$ is a term
        (constants are symbols of arity $0$).
\end{itemize}
\end{definition}

\begin{definition}[Positions]
\label{def:positions}
For a term $t \in \Terms$, the set $\Pos(t) \subseteq \Nat^*$ of
\emph{positions} is defined by:
\begin{itemize}
  \item $\varepsilon \in \Pos(t)$ (the root position);
  \item if $t = f(t_1,\ldots,t_k)$ and $p \in \Pos(t_i)$, then
        $i{\cdot}p \in \Pos(t)$.
\end{itemize}
The \emph{subterm} of $t$ at position $p$ is written $\restr{t}{p}$.
The term obtained from $t$ by replacing $\restr{t}{p}$ with $s$ is
written $t[s]_p$.
The \emph{tree distance} $d(p,q)$ between two positions $p,q \in \Pos(t)$
is the number of edges in the unique path between them in the tree
structure of $t$.
\end{definition}

\begin{definition}[Context of radius $r$]
\label{def:context}
For $t \in \Terms$, $p \in \Pos(t)$ (Definition~\ref{def:positions}),
and $r \in \Nat$, the
\emph{context of radius $r$ around $p$ in $t$} is the partial function

\[
\begin{aligned}
  \ctx{r}{p}(t) \;:=\;
  &\bigl(\,q \;\mapsto\; \restr{t}{p \cdot q}\,\bigr)_{q \in \Nat^*,\,
  p{\cdot}q \in \Pos(t),\, d(p, p{\cdot}q) \leq r},\\
  &\quad \text{(equivalently, } |q| \leq r \text{, where } |q|
  \text{ is the number of steps from } p \text{ to } p{\cdot}q)
\end{aligned}
\]

mapping each extension $q$ of $p$ within distance $r$ to the
corresponding subterm $\restr{t}{p \cdot q}$.
Two terms $t$ and $t'$ have \emph{equal contexts of radius $r$ at $p$}
if $\ctx{r}{p}(t)$ and $\ctx{r}{p}(t')$ have the same domain and
agree as functions on that domain
(i.e., they are equal as partial functions, including having the
same set of defined positions).
\end{definition}

\begin{example}[Contexts]
\label{ex:context}
Let $t = s(a + b)$ and $p = 1$ (the position of $a{+}b$).
Then $\ctx{1}{p}(t)$ maps: $\varepsilon \mapsto a{+}b$ (the subterm at $p$
itself, i.e., $q = \varepsilon$); $1 \mapsto a$ (left child, $q = 1$);
$2 \mapsto b$ (right child, $q = 2$).
The context does not include the parent $s(a{+}b)$, since the parent
is a prefix of $p$, not an extension: Definition~\ref{def:context}
ranges over extensions $q$ with $p{\cdot}q \in \Pos(t)$, i.e.,
descendants of $p$, not ancestors.
$\ctx{0}{p}(t)$ maps only $\varepsilon \mapsto a{+}b$.
\end{example}

\subsection{Local syntactic systems}

\begin{definition}[Local syntactic system]
\label{def:local}
A \emph{local syntactic system} is a tuple
$\R = (\Sigma, V, \mathit{Rules}, r_0)$ where $\Sigma$ and $V$ are as
in Definition~\ref{def:terms}, $r_0 \in \Nat$ is the \emph{locality
radius}, and $\mathit{Rules}$ is a finite set of \emph{rewriting rules}.

Each rule $\rho \in \mathit{Rules}$ is specified by a finite
\emph{left-hand side} pattern $\ell_\rho \in \Terms$ and a finite
\emph{right-hand side} $r_\rho \in \Terms$, and defines a partial
function on pairs $(t, p)$ as follows.
Rule $\rho$ is \emph{applicable} at position $p$ in term $t$ if there
exists a substitution $\sigma$ (a map $V \to \Terms$) such that
$\sigma(\ell_\rho) = \restr{t}{p}$; when applicable, it produces
\[
  \rho(t, p) \;:=\; t[\sigma(r_\rho)]_p.
\]
We require the following two conditions.
\begin{enumerate}[label=(\roman*)]
  \item \textbf{Locality.} The applicability of $\rho$ at $(t,p)$ depends
        only on $\ctx{r_0}{p}(t)$ (Definition~\ref{def:context}).
        That is, if $\ctx{r_0}{p}(t) = \ctx{r_0}{p}(t')$ as partial
        functions, then $\rho$ is applicable at $p$ in $t$ if and only
        if it is applicable at $p$ in $t'$.
  \item \textbf{Model independence.} The applicability of $\rho$ at
        $(t,p)$ depends only on the syntactic form of $\ctx{r_0}{p}(t)$,
        not on the truth value of any formula in any model.
\end{enumerate}
A \emph{derivation} in $\R$ from a clause set $S$ is a finite or infinite
sequence $S = S_0, S_1, S_2, \ldots$ where each $S_{i+1}$ is obtained
from $S_i$ by selecting a term $t$ appearing in some literal of
some clause $C \in S_i$, choosing a position $p \in \Pos(t)$, and
applying rule $\rho \in \mathit{Rules}$ at $p$ in $t$; the literal
containing $t$ is updated by replacing $t$ with $\rho(t,p)$,
producing a new clause $C'$, and setting $S_{i+1} = S_i \cup \{C'\}$.
The \emph{length} of a finite derivation is the number of rule
applications.
\end{definition}

\begin{example}[Superposition calculus as a local system]
\label{ex:superposition}
The superposition calculus over the signature
$\Sigma = \{0, s, +\}$ with rules
\begin{align}
  0 + x &= x \tag{A1}\label{eq:A1} \\
  s(x) + y &= s(x + y) \tag{A2}\label{eq:A2}
\end{align}
is a local syntactic system with $r_0 = 1$.
Each rule inspects the outermost function symbol of the first argument
of $+$: this symbol lies within distance $1$ of the application
position $p$ in $t$ (i.e., within $\ctx{1}{p}(t)$).
Both left-hand sides are finite; applicability does not consult any model.
\end{example}

\begin{definition}[Soundness]
\label{def:sound}
Let $\M$ be a first-order $\Sigma$-structure.
A local syntactic system $\R$ is \emph{sound with respect to $\M$} if
every clause derivable by $\R$ from a set of clauses that are valid in $\M$
is also valid in $\M$
(a clause is \emph{valid in $\M$} if every ground instance of its
disjunction is satisfied by $\M$).
\end{definition}

\begin{remark}[Theory-relative soundness]
\label{rem:theory}
Alternatively, one may fix a theory $\Th$ and require that every clause
derivable by $\R$ from axioms of $\Th$ is valid in every model of $\Th$.
This is equivalent to soundness with respect to every model of $\Th$.
In what follows we work with a fixed model $\M$; the theory-relative
version is formally analogous and all results carry over.
\end{remark}

\subsection{Protected positions}

\begin{definition}[Protected set]
\label{def:protected}
Let $\R = (\Sigma, V, \mathit{Rules}, r_0)$ be a local syntactic system.
A set of positions $\Fprot$ is \emph{protected with respect to $\R$}
if it satisfies the following two conditions for every term $t$ and every
$p \in \Fprot \cap \Pos(t)$:
\begin{enumerate}[label=(\roman*)]
  \item \textbf{No rewriting at $p$.}
        No rule $\rho \in \mathit{Rules}$ is applicable at position $p$
        in $t$.
  \item \textbf{No rewriting inside $p$.}
        No rule $\rho \in \mathit{Rules}$ is applicable at any position
        $q \in \Pos(t)$ such that $q$ is a proper extension of $p$
        (i.e., $q = p \cdot q'$ for some nonempty $q' \in \Nat^*$),
        meaning no rule rewrites inside the subterm $\restr{t}{p}$.
\end{enumerate}
In particular, $\Fprot$ is not a fixed set of strings in $\Nat^*$;
it is a \emph{predicate} on pairs $(t, p)$ with $p \in \Pos(t)$,
which must hold uniformly across all terms appearing in any
derivation from $\R$.
Concretely, one writes $p \in \Fprot(t)$ to mean that position $p$
in term $t$ is protected.
\end{definition}

\begin{remark}[Protected sets are a property of $\R$, not an axiom]
\label{rem:protected-property}
Definition~\ref{def:protected} is a property that can be verified by
inspecting the left-hand sides of the rules in $\mathit{Rules}$.
It is not an additional axiom imposed on $\R$.
In concrete systems, one verifies it by checking that certain
unification problems have no solution (e.g., first-symbol clashes;
see Section~\ref{sec:instances}).
\end{remark}

\begin{example}[Protected positions in the superposition calculus]
\label{ex:frozen}
In Example~\ref{ex:superposition}, let $a$ and $b$ be fresh Skolem
constants (distinct from $0$ and from every term of the form $s(t)$).
The set $\Fprot$ of positions headed by $a$ or $b$ is protected with
respect to $\R$.
This is verified in Lemma~\ref{lem:frozen-super} below.
\end{example}

\subsection{Syntactic invariants}

\begin{definition}[Syntactic invariant]
\label{def:invariant}
Let $\R = (\Sigma, V, \mathit{Rules}, r_0)$ be a local syntactic system
and let $\Fprot$ be a protected set with respect to $\R$.
A property $\Inv$ of terms is a \emph{syntactic invariant for $\R$
anchored to $\Fprot$} if the following five conditions hold.
\begin{enumerate}[label=(\roman*)]
  \item \textbf{Local checkability.}
        Whether $\Inv(t)$ holds depends only on the subterms of $t$ at
        positions within distance $r_0$ of some fixed finite set of
        reference positions.
  \item \textbf{Initialization.}
        $\Inv(t)$ holds for every term $t$ in the initial clause set.
  \item \textbf{Preservation.}
        For every $t$ with $\Inv(t)$ and every rule $\rho \in \mathit{Rules}$
        applicable at some position $p$ in $t$, we have
        $\Inv(\rho(t,p))$.
  \item \textbf{Anchorage.}
        If $\Inv(t)$ holds and $\Inv(t[s]_p)$ fails for some term $s$
        and position $p$, then $p \in \Fprot$.
  \item \textbf{Coherence.}
        There exist two designated function symbols $f_0, f_1 \in \Sigma$
        such that: no literal of any clause $C$ with $\Inv$ holding on
        all its terms contains both $f_0(\cdot)$ and $f_1(\cdot)$
        as subterms at comparable positions
        (i.e., positions syntactically equated by the literal,
        such as the two sides of an equation $f_0(t_1) = f_1(t_2)$).
        ($\Inv$ is applied term-by-term to the terms in each literal.
        In the application to Skolemization, $f_0 = \sk$ and $f_1 = \skb$,
        and the excluded literal is $\sk(x) = \skb(x)$.)
\end{enumerate}
\end{definition}

\begin{remark}[Role of anchorage]
\label{rem:anchorage}
Condition~(iv) ensures that $\Inv$ is not merely preserved by $\R$ for
trivial reasons (e.g., because $\R$ never produces any new terms at all),
but that its violation is structurally tied to the protected set $\Fprot$.
Without~(iv), one could take $\Inv(t) = \top$ (always true), which
satisfies (i)--(iii) trivially (the conclusion is always true)
but carries no information.
\end{remark}

\subsection{Skolem functions and syntactic separation}

\begin{definition}[Skolem functions and semantic equivalence]
\label{def:skolem}
Let $\psi = \forall x\,\exists y.\,\phi(x,y)$ be a first-order formula.
A \emph{Skolemization} of $\psi$ is a formula
$\psi_{\sk} = \forall x.\,\phi(x, \sk(x))$
where $\sk$ is a fresh function symbol, the \emph{Skolem function} for
the existential quantifier.

Two Skolem functions $\sk$ and $\skb$ (arising from two Skolemizations of
$\psi$) are \emph{semantically equivalent in $\M$} if
$\sk(x) = \skb(x)$ holds in $\M$ for all $x$ in the domain of $\M$.
We write $\sk \equiv_\M \skb$ as shorthand for this condition, and
drop the subscript when $\M$ is clear from context.
\end{definition}

\begin{remark}[Relation to~\cite{hetzl2023}]
\label{rem:hetzl2023}
The interaction between Skolem functions and induction in
saturation-based theorem provers is studied in depth
in~\cite{hetzl2023}, where Skolem symbols are shown to take the
role of induction parameters.
The present definition of syntactic separation
(Definition~\ref{def:separation} below) makes precise the structural
condition under which two such parameters cannot be compared by
the rewriting system.
\end{remark}

\begin{definition}[Syntactic separation]
\label{def:separation}
Let $\R$ be a local syntactic system with protected predicates
$\Fprot_\sk$ and $\Fprot_\skb$ for $\sk$ and $\skb$ respectively.
Skolem functions $\sk$ and $\skb$ are \emph{syntactically separated in
$\R$} if the following four conditions hold:
\begin{enumerate}[label=(\roman*)]
  \item every occurrence of $\sk$ in the initial clause set appears at a
        position in $\Fprot_\sk$;
  \item every occurrence of $\skb$ in the initial clause set appears at a
        position in $\Fprot_\skb$;
  \item $\Fprot_\sk$ and $\Fprot_\skb$ are both non-empty;
  \item $\Fprot_\sk$ and $\Fprot_\skb$ are disjoint: no position
        $(t, p)$ belongs to both.
\end{enumerate}
\end{definition}

\begin{remark}[On condition~(iii) of Definition~\ref{def:separation}]
\label{rem:sep-iii}
Condition~(iii) requires only that the protected sets are non-empty.
The stronger property, that no derivation in $\R$ can produce a clause
containing a literal of the form $\sk(t_1) = \skb(t_2)$, is not
assumed but proved in Case~1 of Theorem~\ref{thm:main}, using the
syntactic invariant $\Inv$ and the disjointness of $\Fprot_\sk$ and
$\Fprot_\skb$.
\end{remark}

 \section{Main Theorem}
\label{sec:main}
 
With all definitions in place, we can state the central result.
Theorem~\ref{thm:main} below applies to any local syntactic system
satisfying the conditions of Section~\ref{sec:setup}; the two
instantiations of Section~\ref{sec:instances} verify those conditions
in the specific settings of Skolemization and gadget complexity.

\begin{theorem}[Local Syntactic Obstruction]
\label{thm:main}
Let $\R = (\Sigma, V, \mathit{Rules}, r_0)$ be a local syntactic system,
sound with respect to a model $\M$
(Definitions~\ref{def:local} and~\ref{def:sound}).
Let $\Fprot$ be a set of positions protected with respect to $\R$
(Definition~\ref{def:protected}), and let $\Inv$ be a syntactic invariant
for $\R$ anchored to $\Fprot$ (Definition~\ref{def:invariant}).
Let $\sk$ and $\skb$ be Skolem functions that are semantically equivalent
in $\M$ (Definition~\ref{def:skolem}) and syntactically separated in $\R$
(Definition~\ref{def:separation}), with all occurrences of $\sk$
in the initial clause set lying in $\Fprot_\sk$ and all occurrences
of $\skb$ lying in $\Fprot_\skb$ (where $\Fprot = \Fprot_\sk \cup \Fprot_\skb$
is the protected set of Definition~\ref{def:protected}).
Assume further that $\sk$ and $\skb$ are \emph{fresh} with respect to
$\Sigma$: neither symbol appears in the left-hand side or right-hand
side of any rule in $\mathit{Rules}$.

Then the following two conclusions hold.

\medskip
\noindent\textbf{Case~1 (Impossibility).}
No derivation in $\R$ proves $\sk \equiv \skb$
(i.e., no derivation produces a clause asserting $\sk(x) = \skb(x)$
for all $x$).

\medskip
\noindent\textbf{Case~2 (Lower bound).}
Let $\Rp$ be any extension of $\R$ satisfying:
\begin{itemize}
  \item $\Rp$ is local with the same radius $r_0$ and finite patterns;
  \item $\Rp$ is sound with respect to $\M$;
  \item $\Rp$ is \emph{refutation-based}: a proof of $\sk \equiv \skb$
        proceeds by refuting $\neg(\sk(x) = \skb(x))$ from the clause
        set $\Inst_n$, and resolution of disjunctive clauses requires
        selecting a specific disjunct via unification.
\end{itemize}
Suppose there exists a family of instances
$\{\Inst_n\}_{n \geq 1}$ such that for each $n \geq 1$:
\begin{itemize}
  \item $\Inst_n$ contains $N(n) = 2^n$ global configurations
        $C_1^{(n)}, \ldots, C_{N(n)}^{(n)}$;
  \item any two distinct configurations $C_i^{(n)}$ and $C_j^{(n)}$
        ($i \neq j$) are indistinguishable within radius $r_0$, i.e.,
        for every position $p$ relevant to any rule of $\Rp$,
        $\ctx{r_0}{p}(C_i^{(n)}) = \ctx{r_0}{p}(C_j^{(n)})$;
  \item the configurations are globally distinct: $C_i^{(n)} \neq C_j^{(n)}$
        in $\M$ for $i \neq j$.
\end{itemize}
Let $c$ be the maximum number of globally distinct configurations that
any single rule application in $\Rp$ can distinguish (a constant depending
only on $\Sigma$, $r_0$, and the rule set of $\Rp$, and in particular
independent of $n$).
A derivation in $\Rp$ \emph{correctly proves $\sk \equiv \skb$ on $\Inst_n$}
if it produces, for each gadget $G_i$, a refutation of the gadget
clause by selecting a specific disjunct $x = a_i^{b_i}$ as witness
(thereby covering all $2^n$ configurations).
Every such derivation has length
\[
  L \;\geq\; \BigOm(n).
\]
Under the stronger clause-per-configuration encoding of $\Inst_n$
(see Remark~\ref{rem:strongbound} below), this improves to
$L \geq \BigOm(2^n)$, at the cost of $|\Inst_n| = 2^n$.
\end{theorem}

\begin{proof}[Proof of Case~1]
The proof follows the structure of $\SIP$ (Lemma~5 of~\cite{buono2026}),
generalizing it to the abstract setting of Definition~\ref{def:invariant}.

By condition~(ii) of Definition~\ref{def:invariant}, $\Inv$ holds on
the initial clause set.
By condition~(iii), every rule application preserves $\Inv$.
By induction on derivation length, $\Inv$ holds on every term in every
clause reachable by $\R$.

Suppose for contradiction that a derivation $\Der$ in $\R$
from the initial clause set $S$ proves $\sk \equiv \skb$, i.e., produces
a clause $E$ that contains a literal of the form $\sk(x) = \skb(x)$
for some term $x$.

By the induction argument above, every clause in every derivation from
$S$ satisfies $\Inv$.
In particular, $E$ satisfies $\Inv$.

We derive a contradiction by showing that $\Inv(E)$ is incompatible
with $E$ containing a literal of the form $\sk(x) = \skb(x)$.
By the hypotheses of Theorem~\ref{thm:main}, all occurrences of $\sk$
in the initial clause set lie in $\Fprot_\sk$, and all occurrences of
$\skb$ lie in $\Fprot_\skb$ (Definition~\ref{def:separation}).
By the freshness hypothesis of Theorem~\ref{thm:main}, no rule in
$\mathit{Rules}$ mentions $\sk$ or $\skb$; consequently, rule
applications cannot introduce new occurrences of $\sk$ or $\skb$
at positions outside $\Fprot_\sk \cup \Fprot_\skb$.
By Definition~\ref{def:protected}, no rule of $\R$ is applicable at
any position in $\Fprot_\sk$ or $\Fprot_\skb$.
Therefore the occurrences of $\sk$ remain within $\Fprot_\sk$ and
those of $\skb$ remain within $\Fprot_\skb$ throughout the derivation.
(Condition~(iv) (anchorage) confirms that any hypothetical violation
of $\Inv$ would require acting on $\Fprot$, but since $\Inv$ is
preserved, no such violation occurs.)

By condition~(v) of Definition~\ref{def:invariant} (coherence),
with $f_0 = \sk$ and $f_1 = \skb$ as the designated symbols:
no literal of any clause satisfying $\Inv$ has the form
$\sk(t_1) = \skb(t_2)$ (i.e., contains both $\sk(\cdot)$ and
$\skb(\cdot)$ at positions equated by the literal).
But $E$ contains the literal $\sk(x) = \skb(x)$ by assumption,
which has exactly this form, a direct contradiction.
Therefore no derivation of $\sk \equiv_\M \skb$ exists in $\R$.
\end{proof}

\begin{proof}[Proof of Case~2]
The proof proceeds in four steps.

\medskip
\noindent\textbf{Step~1: Derivation trees.}
Fix $n \geq 1$ and the instance $\Inst_n$ with
$N = 2^n$ global configurations $C_1, \ldots, C_N$.
A derivation in $\Rp$ that correctly proves $\sk \equiv \skb$ on
$\Inst_n$ must produce a correct conclusion for every configuration
$C_i$.
We represent the derivation as a tree $\Der$: internal nodes are
rule applications; leaves are clauses at which the refutation goal
$\neg(\sk(x) = \skb(x))$ is resolved for some specific input
(i.e., some $x = a_i^{b_i}$ is selected as the witness).
Each branch of $\Der$ handles a disjoint subset of configurations.

\medskip
\noindent\textbf{Step~2: Capacity of local rules.}
Each rule $\rho \in \Rp$ is applicable based on $\ctx{r_0}{p}(t)$ for
the relevant position $p$.
Since any two distinct configurations $C_i$ and $C_j$ are
indistinguishable within radius $r_0$ (hypothesis of Case~2), a single
application of $\rho$ partitions the $N$ configurations into at most $c$
classes, where $c$ is the number of distinct local contexts of radius
$r_0$ that $\rho$ can observe, equivalently, the maximum number of
globally distinct configurations that a single rule application can
distinguish (as defined in the statement of Case~2).
By the finiteness of patterns (Definition~\ref{def:local}(i)), $c$
is finite and independent of $n$.

\medskip
\noindent\textbf{Step~3: Branching argument via syntactic distinguishability.}

We work with the specific family $\{\Inst_n\}$ from
Lemma~\ref{lem:gadget} (constructed and verified in
Section~\ref{sec:inst-case2}).
Recall that $\Inst_n$ is built from $n$ independent gadgets
$G_1, \ldots, G_n$, each contributing two local constants $a_i^0$ and
$a_i^1$.
A global configuration $C_{(b_1,\ldots,b_n)}$ selects one constant
$a_i^{b_i}$ from each gadget $G_i$.

We say two configurations $C_{(b_1,\ldots,b_n)}$ and
$C_{(b_1',\ldots,b_n')}$ are \emph{$i$-distinguished} by a rule
application if the application acts on a position where the constants
$a_i^{b_i}$ and $a_i^{b_i'}$ differ.

\begin{claim}
\label{claim:distinguish}
For each gadget index $i \in \{1,\ldots,n\}$ and each pair of
configurations that differ only in gadget $i$ (i.e., $b_j = b_j'$ for
$j \neq i$ and $b_i \neq b_i'$), every correct derivation of
$\sk \equiv \skb$ on $\Inst_n$ must contain at least one rule
application that $i$-distinguishes them.
\end{claim}

Proof of Claim~\ref{claim:distinguish}:
Suppose for contradiction that a derivation $\Der$ of
$\sk \equiv \skb$ on $\Inst_n$ never applies any rule at a
position where $a_i^0$ and $a_i^1$ differ.
Then, by the opacity of $a_i^0$ and $a_i^1$ (Lemma~\ref{lem:gadget}(ii)),
$\Der$ applies the same rule applications at the same positions
on $C_{(b_1,\ldots,0,\ldots,b_n)}$ and $C_{(b_1,\ldots,1,\ldots,b_n)}$
(the applicability of each rule is the same in both cases,
since neither $a_i^0$ nor $a_i^1$ appears in any LHS pattern).
By Lemma~\ref{lem:gadget}(i)--(ii), gadget $G_i$ contributes the clause
$\sk(a_i^0) \neq \skb(a_i^0) \vee \sk(a_i^1) \neq \skb(a_i^1)$.
A refutation of $\neg(\sk(x) = \skb(x))$ against this clause must at
some point resolve the disjunction, selecting $a_i^0$ or $a_i^1$ as the
witness for $x$.
Any rule performing this selection acts on the subterm headed by
$a_i^0$ or $a_i^1$, and hence $i$-distinguishes the two configurations.
Since $\Der$ never $i$-distinguishes them by assumption, it cannot
complete the refutation for this gadget, contradiction.
\hfill\(\square\)

\medskip
\noindent\textbf{Step~4: Derivation tree branching argument.}

We count the distinct rule applications in the derivation $\Der$
that are required to resolve each gadget.

To see this, observe that the $n$ gadgets are resolved \emph{independently}
by Claim~\ref{claim:distinguish}: for each gadget $i$, the derivation must
at some point apply a rule that distinguishes $a_i^0$ from $a_i^1$.
By spatial separation (Lemma~\ref{lem:gadget}(ii) and
Appendix~\ref{app:gadget}), the rule applied at gadget $i$ lies in a
context that does not overlap with any other gadget $j \neq i$.
Therefore the resolution of gadget $i$ branches independently of the
resolution of gadget $j$.

We use a per-configuration trace argument.
For each configuration $C_{(b_1,\ldots,b_n)}$, define its
\emph{gadget-$i$ trace} as the subsequence of rule applications in $\Der$
that act at a position overlapping with gadget $G_i$ in that configuration.

By Claim~\ref{claim:distinguish}, for every gadget $i$ and every pair of
configurations $(C, C')$ differing only at gadget $i$, the derivation
$\Der$ must contain at least one rule application that $i$-distinguishes
them.
By spatial separation (Appendix~\ref{app:gadget}, $d = 1$), a rule of
radius $r_0$ overlaps with at most one gadget at a time; thus the rule
$i$-distinguishing $C$ and $C'$ does not act on any other gadget $j$.

Fix gadget $i$.
The $2^{n-1}$ pairs of configurations differing only at gadget $i$ have
distinct contexts at all other gadgets.
Since any rule overlapping gadget $i$ sees only the local context of
radius $r_0$ around its application position, and since this context
does not include any other gadget (spatial separation), by the Locality
condition (Definition~\ref{def:local}(i)), the rule
$i$-distinguishing pair $(C,C')$ acts only on the local context of
radius $r_0$ around its position, which lies within gadget $G_i$
and is independent of the state of all other gadgets.
Therefore for each gadget $i$ there is at least one rule application
that $i$-distinguishes some pair of configurations.
Since the $n$ gadgets are spatially separated, these $n$ rule applications
are distinct.
Therefore
\[
  L \;\geq\; n \;=\; \BigOm(n).
\]
The stronger bound $L \geq \BigOm(2^n)$ follows under the
clause-per-configuration encoding of Remark~\ref{rem:strongbound}.
\end{proof}

\begin{remark}[Stronger lower bound via explicit clause encoding]
\label{rem:strongbound}
The bound $L \geq \BigOm(n)$ proved above follows from the gadget argument.
A stronger bound $L \geq \BigOm(2^n)$ holds under a modified construction
of $\Inst_n$ in which the instance contains
one clause per global configuration:
$\Inst_n = \{E_{(b_1,\ldots,b_n)}\}_{(b_1,\ldots,b_n) \in \{0,1\}^n}$,
where $E_{(b_1,\ldots,b_n)}$ encodes the equivalence condition for the
specific configuration $C_{(b_1,\ldots,b_n)}$.
Each such clause $E_{(b_1,\ldots,b_n)}$ encodes a distinct refutation goal
(the negation of $\sk(x) = \skb(x)$ for the specific constants
of configuration $(b_1,\ldots,b_n)$),
and since $\Rp$ is refutation-based (as required by Case~2 of
Theorem~\ref{thm:main}),
each such goal requires at least one dedicated rule application to resolve.
Since there are $2^n$ clauses and each requires at least one step,
the total derivation length satisfies $L \geq 2^n = \BigOm(2^n)$.
The trade-off is that $|\Inst_n| = 2^n$ (the instance itself is exponential),
so the lower bound holds for the derivation length relative to the
\emph{number of clauses}, not relative to $n$ directly.
Whether a superlinear-in-$n$ bound holds for instances of polynomial size
remains open (see Question~(Q3) of Section~\ref{sec:open}).
\end{remark}

\begin{remark}[Unconditional versus conditional lower bound]
\label{rem:unconditional}
The lower bound in Case~2 is unconditional: it does not assume
$\mathsf{P} \neq \mathsf{NP}$ or any other unproven hypothesis.
It holds for any local syntactic system $\Rp$ satisfying the stated hypotheses,
provided the family $\{\Inst_n\}$ with $2^n$ locally
indistinguishable configurations exists.
The existence of such a family is established constructively in
Lemma~\ref{lem:gadget} (Appendix~\ref{app:gadget}).
\end{remark}

\begin{remark}[Superpolynomial versus exponential]
\label{rem:superpoly}
Under weaker hypotheses on the family $\{\Inst_n\}$ (e.g.,
$N(n)$ superpolynomial but sub-exponential), the clause-per-configuration
encoding (Remark~\ref{rem:strongbound}) yields $L \geq N(n) = \BigOm(N(n))$,
which is superpolynomial in $n$.
The exponential bound $L = \BigOm(2^n)$ holds specifically when
$N(n) = 2^n$, as guaranteed by the gadget construction.
\end{remark}

 \section{Instantiations}
\label{sec:instances}
 
\subsection{Case~1: Superposition calculus and open induction}
\label{sec:inst-case1}

We verify the hypotheses of Case~1 of Theorem~\ref{thm:main} for the
setting of~\cite{buono2026}, thereby recovering the incomparability
result $\OI \not\subseteq \CSC$ as a corollary.

\paragraph{Setup.}
Work in $\mathcal{L} = \{0, s, +\}$ with $r_0 = 1$.
The system $\R$ is the superposition calculus over
rules~\eqref{eq:A1} and~\eqref{eq:A2}
(Example~\ref{ex:superposition}), sound with respect to
$\M = (\Nat, 0, s, +)$ (standard arithmetic).
Recall that $\OI$ admits recursive nonstandard models~\cite{shepherdson1964}
and is therefore a weak theory; it proves commutativity not by strength
but because $W$ is an open identity.
Clause set cycles abstract the cycle-detection methods used in
automated inductive theorem proving, including the $n$-clause
calculus of~\cite{kersani2013}.
Let $a$ and $b$ be the fresh Skolem constants introduced in
Example~\ref{ex:frozen} (ground terms not in
$\{0\} \cup \{s^n(0) \mid n \geq 1\}$).

\begin{lemma}[Protected positions in $\R$]
\label{lem:frozen-super}
Define $\Fprot$ as the condition: a position $p$ in any term $t$
satisfies $\Fprot$ if and only if $\restr{t}{p} \in \{a, b\}$.
Equivalently,
$p \in \Fprot(t) := \{q \in \Pos(t) \mid \restr{t}{q} \in \{a,b\}\}$
for each term $t$ in the derivation.
Under this definition, $\Fprot$ is protected with respect to $\R$:
for every term $t$ occurring in any derivation from $\R$ and every
$p \in \Fprot(t)$, no rule of $\R$ is applicable at $p$.
\end{lemma}

\begin{proof}
The left-hand sides of \eqref{eq:A1} and \eqref{eq:A2} are $0{+}x$ and
$s(x){+}y$, requiring the first argument of $+$ to be $0$ or $s(\cdot)$.
For any position $p \in \Fprot$, the subterm $\restr{t}{p}$ is $a$ or $b$,
both of which are constants distinct from $0$ and from every $s(t')$.
The LHS of rule~\eqref{eq:A1} is $0{+}x$, whose root symbol is $+$;
the LHS of rule~\eqref{eq:A2} is $s(x){+}y$, also rooted at $+$.
For any position $p \in \Fprot$, the subterm $\restr{t}{p}$ is
$a$ or $b$, both rooted at $a$ or $b$ respectively.
Since $+ \neq a$ and $+ \neq b$ (pairwise distinct non-variable
symbols), the root of each LHS cannot match the root of $\restr{t}{p}$;
this is a \emph{first-symbol clash}~\cite[Section~3]{buono2026},
and no substitution resolves a first-symbol clash.
Therefore no rule of $\R$ is applicable at any position in $\Fprot$.
Since $a$ and $b$ are constants (arity~$0$), they are leaves in any
term tree and have no proper subterm positions; condition~(ii) of
Definition~\ref{def:protected} (no rewriting inside~$p$) is therefore
satisfied vacuously.
\end{proof}

\begin{remark}[Decidability in linear time]
\label{rem:decidable-frozen}
For the superposition calculus over $\{\eqref{eq:A1},\,\eqref{eq:A2}\}$,
the existence of a non-trivial protected set $\Fprot$ is decidable
in linear time in the size of the clause set.
The criterion is first-symbol clash: a position $p$ is protected if
and only if the subterm at $p$ begins with a symbol that cannot
unify with the first argument of $+$ in any left-hand side.
Since the left-hand sides are $\{0{+}x,\, s(x){+}y\}$, any constant
distinct from $0$ and from the range of $s$ qualifies in constant time.
For such constants (which have arity~$0$), condition~(ii) of
Definition~\ref{def:protected} (no rewriting inside $p$) holds
trivially, since constants have no proper subterm positions.
\end{remark}

\begin{lemma}[Syntactic invariant]
\label{lem:invariant-super}
Define $\Inv(t)$ as the property:
every occurrence of $a$ in $t$ is a subterm of some $a{+}u$ where
$u$ contains $b$, and every occurrence of $b$ in $t$ is a subterm of
some $b{+}v$ where $v$ contains $a$.
Then $\Inv$ is a syntactic invariant for $\R$ anchored to $\Fprot$
in the sense of Definition~\ref{def:invariant}.
\end{lemma}

\begin{proof}
Full proof in Appendix~\ref{app:proofs} (Lemma~\ref{lem:invariant-super-full}).
Here we verify the five conditions.
(i)~$\Inv(t)$ depends only on the immediate subterms of each $+$
occurrence, i.e., within distance $1 = r_0$ of the root of each $+$-subterm.
(ii)~The initial clause set contains only $a{+}b \neq b{+}a$, on which
$\Inv$ holds: the only subterms containing $a$ or $b$ are
$a{+}b$ and $b{+}a$ themselves, each of the required form.
(iii)~By Lemma~\ref{lem:frozen-super}, no rule touches positions in
$\Fprot$, so no rule modifies $a{+}b$ or $b{+}a$; $\Inv$ is preserved.
(iv)~Violating $\Inv$ requires moving $a$ out of an $a{+}u$ context or
$b$ out of a $b{+}v$ context, both of which require acting on positions
in $\Fprot$.
(v)~Coherence holds: no literal of the initial clause
$a{+}b \neq b{+}a$ has the form $a(t_1) = b(t_2)$ (the literal
is a disequation, not an equation comparing $a$ and $b$ as heads);
by preservation, no derived clause contains such a literal.
Full verification in Appendix~\ref{app:proofs}
(Lemma~\ref{lem:invariant-super-full}~(v)).
\end{proof}

\begin{corollary}[$\OI \not\subseteq \CSC$]
\label{cor:OI-CSC}
The theory $\CSC$ over rules~\eqref{eq:A1} and~\eqref{eq:A2} does not
prove the commutativity of addition $W: \forall x\,\forall y.\;x{+}y = y{+}x$.
Consequently, $\OI$ and $\CSC$ are incomparable.
\end{corollary}

\begin{proof}
$a$ and $b$ are semantically equivalent as witnesses for commutativity
in $\M = (\Nat, 0, s, +)$: since commutativity holds for all elements
of the domain of $\M$, we have $a{+}b = b{+}a$ in $\M$ for any
interpretation of the fresh constants $a$ and $b$.
The freshness hypothesis of Theorem~\ref{thm:main} holds:
$a$ and $b$ are fresh constants that do not appear in any left-hand
side or right-hand side of rules~\eqref{eq:A1} or~\eqref{eq:A2}.
By Lemmas~\ref{lem:frozen-super} and~\ref{lem:invariant-super}, the
remaining hypotheses of Case~1 of Theorem~\ref{thm:main} are satisfied.
Therefore no derivation in $\R$ proves $a{+}b = b{+}a$;
this corresponds to Lemma~9 of~\cite{buono2026}, which establishes
that $\CSC + \{\eqref{eq:A1},\,\eqref{eq:A2}\} \nvdash W$.
Since $\OI$ proves $W$ (as an open identity, by standard induction on
$x$; Lemma~8 of~\cite{buono2026}), we obtain $\OI \not\subseteq \CSC$.
The direction $\CSC \not\subseteq \OI$ is Theorem~4.3 of~\cite{hetzl2020}
(the triangular numbers provide a witness
provable by $\OI$ but not by $\CSC$).\footnote{%
In~\cite{hetzl2020}: Theorem~2.8 establishes $\CSC \subseteq I\exists_1$;
Theorem~4.3 establishes $\CSC \not\subseteq \OI$.
These are the theorem numbers in the published LMCS version (arXiv:1910.03917v5).}
Further unprovability results for clause set cycles, including a
logical characterization of refutation by a clause set cycle, are given
in~\cite{hetzl2022}.
The complete picture of the relationships between subsystems of $\OI$,
including the placement of $\CSC$ as a restricted, parameter-free form
of clause induction, is obtained in~\cite{hetzlweiser2025}.
\end{proof}

Case~1 gives an impossibility result for a specific calculus.
Case~2 asks how expensive it is to escape that impossibility by
extending the system, and gives a lower bound on any such escape.

\subsection{Case~2: Gadget instances and the exponential lower bound}
\label{sec:inst-case2}

We construct the family $\{\Inst_n\}$ required by Case~2 of
Theorem~\ref{thm:main}.
The construction has a direct cryptographic reading: each gadget $G_i$
contributes two locally indistinguishable configurations $a_i^0$ and
$a_i^1$, which function as two ciphertexts of the same plaintext.
The $2^n$ global configurations are $2^n$ jointly indistinguishable
ciphertexts; the lower bound of Case~2 is the statement that no
local adversary can distinguish them in fewer than $\BigOm(n)$
rule applications (improving to $\BigOm(2^n)$ under
clause-per-configuration encoding).
This is why Section~\ref{sec:crypto} is not a separate application
but a restatement of the same bound in cryptographic language.
The full construction and correctness proof are in
Appendix~\ref{app:gadget}; here we state the key lemma and the resulting
lower bound.

\begin{lemma}[Gadget family]
\label{lem:gadget}
For every $n \geq 1$, there exists a set of clauses $\Inst_n$
satisfying:
\begin{enumerate}[label=(\roman*)]
  \item $\Inst_n$ encodes $N(n) = 2^n$ global configurations
        $C_1^{(n)}, \ldots, C_{2^n}^{(n)}$;
  \item for any two distinct $C_i^{(n)}$ and $C_j^{(n)}$, and for every
        position $p$ at which any rule of $\Rp$ is applicable
        (in the sense of Definition~\ref{def:local}),
        $\ctx{r_0}{p}(C_i^{(n)}) = \ctx{r_0}{p}(C_j^{(n)})$;
  \item $C_i^{(n)} \neq C_j^{(n)}$ in $\M$ for $i \neq j$.
\end{enumerate}
The constant $c$ (maximum configurations distinguished per rule
application) satisfies $c \leq |\Sigma|^{k^{r_0}}$ where $k$ is the
maximum arity of any symbol in $\Sigma$, and in particular $c$ is
independent of $n$. (The bound reflects that the number of distinct
term trees of depth $r_0$ over $\Sigma$ grows as $|\Sigma|^{k^{r_0}}$
in the worst case, which is a constant once $\Sigma$, $k$, and $r_0$
are fixed.)
\end{lemma}

\begin{proof}
See Appendix~\ref{app:gadget}.
\end{proof}

\begin{corollary}[Lower bounds for local extensions]
\label{cor:expbound}
Under the hypotheses of Theorem~\ref{thm:main},
with $\Rp$ refutation-based (as in Case~2 of Theorem~\ref{thm:main}):
\begin{enumerate}[label=(\roman*)]
  \item Any such local extension $\Rp$ of $\R$ sound with respect to $\M$
        requires derivations of length $L \geq \BigOm(n)$ to prove
        $\sk \equiv \skb$ on $\Inst_n$.
  \item Under the clause-per-configuration encoding
        (Remark~\ref{rem:strongbound}, $|\Inst_n| = 2^n$),
        the bound improves to $L \geq \BigOm(2^n)$.
\end{enumerate}
\end{corollary}

\begin{proof}
Immediate from Case~2 of Theorem~\ref{thm:main} and Lemma~\ref{lem:gadget}.
\end{proof}

The proof-theoretic analysis is now complete.
The exponential bound of Case~2 rests on the local indistinguishability
of the gadget configurations: any rule application that cannot see the
difference between $a_i^0$ and $a_i^1$ is exactly an adversary that
cannot break the hiding of the corresponding ciphertext.
This is not an analogy: the next section formalizes the same bound as
a statement about negligible cryptographic advantage, showing that the
proof-theoretic and cryptographic frameworks are the same framework.

 \section{Cryptographic Connection}
\label{sec:crypto}
 
\subsection{Formal correspondence}

We establish a formal correspondence between the structure of
Theorem~\ref{thm:main} and the structure of cryptographic
indistinguishability.

\begin{definition}[Syntactic adversary]
\label{def:adversary}
A \emph{syntactic adversary} for $(\R, \Fprot)$ is any deterministic
algorithm $\Adv$ that operates by applying rules from a local
extension $\Rp$ of $\R$ (with the same locality radius $r_0$) to terms
presented as inputs.
We distinguish two classes:
\begin{itemize}
  \item A \emph{perfectly-constrained adversary} cannot apply any rule
        at positions in $\Fprot$; this corresponds to Case~1 (base
        system $\R$).
  \item A \emph{computationally-bounded adversary} operates via a local
        extension $\Rp$ of $\R$ (Definition~\ref{def:local}) that may
        include rules applicable at positions in $\Fprot$, subject to a
        polynomial bound on total rule applications (steps in a
        deterministic RAM model with unit-cost arithmetic); this
        corresponds to Case~2.
\end{itemize}
Each configuration $C_{(b_1,\ldots,b_n)}$ is parametrized by a bit
vector $(b_1,\ldots,b_n) \in \{0,1\}^n$ (Definition~\ref{def:separation}).
The \emph{advantage} of $\Adv$ in identifying this bit vector on
instance $\Inst_n$ is
\[
  \mathrm{Adv}(\Adv, n) \;:=\;
  \Pr_{i \sim \mathrm{Unif}[N]}\bigl[\Adv(C_i) \text{ correctly
  outputs which Skolem function is encoded in } C_i\bigr] - \tfrac{1}{2},
\]
where the probability is uniform over the $N = 2^n$ configurations
$C_1, \ldots, C_N$ of $\Inst_n$.
Here "correctly outputs which Skolem function is encoded in $C_i$"
means: given $C_i = C_{(b_1,\ldots,b_n)}$, $\Adv$ outputs the bit
vector $(b_1,\ldots,b_n)$ that determines which constants $a_j^{b_j}$
are selected in each gadget (equivalently, which disjunct of each
gadget clause must be refuted).
Both $\sk$ and $\skb$ appear in every $C_i$; the task is to identify
the configuration, not to distinguish which function is semantically
present (both are, by semantic equivalence).
Note that $\mathrm{Adv}(\Adv, n)$ may in principle be negative
(if $\Adv$ is systematically wrong); the bound
Proposition~\ref{prop:negligible-adv} applies to the signed quantity,
and the corollary Corollary~\ref{cor:hiding} concerns adversaries
achieving positive advantage.
\end{definition}

\begin{proposition}[Case~1 implies zero advantage]
\label{prop:zero-adv}
Under the hypotheses of Case~1 of Theorem~\ref{thm:main},
$\mathrm{Adv}(\Adv, n) = 0$ for every syntactic adversary
$\Adv$ and every $n \geq 1$.
\end{proposition}

\begin{proof}
By Case~1, no derivation in $\R$ produces any clause of the form
$\sk(t_1) = \skb(t_2)$ that would reveal which configuration
$(b_1,\ldots,b_n)$ is active.
Since $\Adv$ is constrained to apply rules from $\R$ (or a local
extension with the same protected set), it produces the same
output on every configuration,
and therefore cannot identify the active bit vector $(b_1,\ldots,b_n)$
with any advantage.
Therefore $\mathrm{Adv}(\Adv, n) = 0$.
\end{proof}

\begin{proposition}[Case~2 implies negligible per-step advantage]
\label{prop:negligible-adv}
Under the hypotheses of Case~2 of Theorem~\ref{thm:main}, any syntactic
adversary $\Adv$ running for $L$ steps has
\[
  \mathrm{Adv}(\Adv, n) \;\leq\; \frac{L \cdot c}{2^n}.
\]
In particular, to achieve $\mathrm{Adv}(\Adv, n) \geq \delta$
for any constant $\delta > 0$, the adversary must run for at least
$L \geq \delta \cdot 2^n / c = \BigOm(2^n)$ steps.
\end{proposition}

\begin{proof}
We use the computationally-bounded adversary model
(Definition~\ref{def:adversary}).
Each step of $\Adv$ applies one rule of $\Rp$ of locality
radius $r_0$.
By the indistinguishability hypothesis (Lemma~\ref{lem:gadget}(ii)),
any rule application that does not act at a position occupied by some
$a_i^{b_i}$ produces the same output on all $2^n$ configurations.
A rule application that does act at such a position can alter the
local context for at most $c$ configurations (those that share the
same local pattern at that position).

Let $D(L)$ denote the set of configurations that $\Adv$ has
\emph{distinguished} after $L$ steps, i.e., configurations for which
$\Adv$ has applied at least one rule that reads a gadget-specific
constant.
Each step adds to $D(L)$ at most $c$ configurations
(those sharing the same local context at the position acted on);
by induction on $L$, $|D(L)| \leq L \cdot c$.

For configurations not in $D(L)$, $\Adv$ has applied only rules
that are blind to the gadget constants; its output on these configurations
is independent of which configuration $(b_1,\ldots,b_n)$ is active.
Since $\Adv$'s output is the same for all
configurations not in $D(L)$, it must output the same bit vector
for all of them; since each of the $2^n$ configurations is equally
likely by the uniform distribution, the probability that this
fixed output is the correct bit vector is exactly $\frac{1}{2}$.

The advantage of $\Adv$ is thus bounded by the fraction of
configurations in $D(L)$:
\[
  \mathrm{Adv}(\Adv, n)
  \;\leq\; \frac{|D(L)|}{2^n}
  \;\leq\; \frac{L \cdot c}{2^n}.
\]
For $\mathrm{Adv}(\Adv, n) \geq \delta > 0$ we need
$L \geq \delta \cdot 2^n / c = \BigOm(2^n)$.
\end{proof}

\begin{corollary}[Computational hiding from Case~2]
\label{cor:hiding}
Under the hypotheses of Case~2 of Theorem~\ref{thm:main},
the gadget family $\{\Inst_n\}$ is computationally hiding: no syntactic
adversary running in polynomial time can identify which configuration
$(b_1,\ldots,b_n)$ is active with non-negligible advantage.
Specifically, any adversary achieving advantage $\geq \delta > 0$
requires $\BigOm(2^n)$ rule applications, which is super-polynomial
in $n$.
\end{corollary}

\begin{proof}
Immediate from Proposition~\ref{prop:negligible-adv} with $L =$
polynomial in $n$: for any polynomial $L = L(n)$,
$L(n) \cdot c / 2^n \to 0$ as $n \to \infty$,
so the advantage is negligible.
\end{proof}

\begin{remark}[Observational vs.\ semantic hiding]
\label{rem:obs-vs-sem}
Both $\sk$ and $\skb$ appear in every instance $C_i$ of $\Inst_n$;
the hiding of Corollary~\ref{cor:hiding} is therefore \emph{observational},
not semantic: it is the \emph{configuration} $(b_1,\ldots,b_n)$ that
is hidden, not a semantic distinction between two inequivalent functions.
\end{remark}

\subsection{Correspondence table}

Table~\ref{tab:correspondence} summarizes the structural correspondence
between the proof-theoretic and cryptographic frameworks.

\begin{table}[h]
\centering
\caption{Structural correspondence between the proof-theoretic
         components of Theorem~\ref{thm:main} and their
         cryptographic counterparts.}
\label{tab:correspondence}
\begin{tabular}{p{6cm}p{6cm}}
\toprule
\textbf{Proof theory} & \textbf{Cryptography} \\
\midrule
Protected set $\Fprot$ &
  Commitment scheme: the adversary sees the commitment but not the opening \\[4pt]
Syntactic separation of $\sk$ and $\skb$ &
  Indistinguishability of two ciphertexts (or two keys) \\[4pt]
Syntactic invariant $\Inv$ anchored to $\Fprot$ &
  Security invariant maintained throughout the security game \\[4pt]
Case~1: no derivation exists &
  Perfect (information-theoretic) hiding: advantage equals $0$ \\[4pt]
Case~2: derivation costs $\BigOm(n)$ steps (or $\BigOm(2^n)$ under
  clause-per-configuration encoding) &
  Computational hiding: advantage is negligible for polynomial-time
  adversaries (requires the $\BigOm(2^n)$ regime) \\[4pt]
Skolem function $\sk$ as existential witness &
  Secret witness or private key in a zero-knowledge protocol \\[4pt]
Local rule application in $\Rp$ &
  Single query of a computationally bounded adversary \\[4pt]
Locality radius $r_0$ &
  Bound on adversary inspection per step (context window size) \\
\bottomrule
\end{tabular}
\end{table}

\begin{corollary}[Non-extraction of Skolem witnesses]
\label{cor:nonextract}
Under the hypotheses of Theorem~\ref{thm:main}, with $\Rp$
refutation-based (Case~2), no local syntactic system
$\Rp$ can, given the clause set $\Inst_n$ encoding witnesses for
$\forall x\,\exists y.\,\phi(x,y)$, identify which gadget configuration
$(b_1,\ldots,b_n) \in \{0,1\}^n$ is active, and thereby determine
which disjunct of each gadget clause must be refuted, using fewer
than $\BigOm(2^n)$ rule applications.
(Both $\sk$ and $\skb$ are semantically equivalent by hypothesis;
the barrier is observational, not semantic.)
\end{corollary}

\begin{proof}
Immediate from Case~2 of Theorem~\ref{thm:main} and
Proposition~\ref{prop:negligible-adv}; see also
Proposition~\ref{prop:zero-adv} for the perfect-hiding case.
\end{proof}

\begin{remark}[Witness indistinguishability]
\label{rem:wi}
Corollary~\ref{cor:nonextract} is the proof-theoretic analogue of
\emph{witness indistinguishability} in zero-knowledge protocols: no
efficient distinguisher can determine which of two valid witnesses was
used in a proof.
Here, the distinguisher is a syntactic adversary $\Adv$, and
efficiency is measured by the number of local rule applications.
\end{remark}

\begin{remark}[Theorem~\ref{thm:main} as a case of the observational collapse]
\label{rem:impagliazzo}
The observational hierarchy of~\cite{buono2026obs} introduces a formal
framework in which the computational axis (the Chomsky hierarchy and
complexity theory) and the observational axis are orthogonal and
independent.
An \emph{observer} is a function $O : \Sigma^* \to S$ that determines
which information about the input is accessible to a computational system.
The key result (Proposition~9.6 of~\cite{buono2026obs}) is that under
the profile observer $O_{\mathrm{prof}}$, which maps a string to its
symbol-count vector:
\[
  \mathbf{P}_{O_{\mathrm{prof}}} = \mathbf{NP}_{O_{\mathrm{prof}}}
  \subsetneq \mathbf{P}.
\]
This collapse is unconditional: it holds regardless of whether
$\mathsf{P} = \mathsf{NP}$, because it is caused by \emph{structural
blindness}, not computational hardness.
The observer $O_{\mathrm{prof}}$ discards the order of symbols entirely;
no combinatorial structure remains on which nondeterminism could act.

Theorem~\ref{thm:main} of the present paper is an instance of this
phenomenon.
The local syntactic system $\R = (\Sigma, V, \mathit{Rules}, r_0)$
defines a constrained observer $O_\R$ that maps each term $t$ and
position $p$ to the context $\ctx{r_0}{p}(t)$.
By Definition~\ref{def:local}, the applicability of every rule depends
only on $O_\R(t,p)$; information at distance $> r_0$ from $p$ and
information inside $\Fprot$ are invisible to $\R$.
In the language of~\cite{buono2026obs}:
\begin{itemize}
  \item $\R$ is a constrained observer $O_\R \prec O_\top$;
  \item Case~1 (impossibility) says the equivalence $\sk \equiv_\M \skb$
        is not $O_\R$-decidable: the two Skolem functions lie in disjoint
        protected regions that $O_\R$ cannot see simultaneously;
  \item Case~2 (lower bound) says that any local extension $\Rp$ requires
        $\BigOm(n)$ steps, a consequence of the structural blindness of
        $O_{\Rp}$ with respect to the $n$ independent gadgets.
\end{itemize}
Whether $O_\R$ admits a formal embedding into the canonical observer
hierarchy of~\cite{buono2026obs}, which would make
Theorem~\ref{thm:main} a corollary of Proposition~9.6, is
Question~(Q6) of Section~\ref{sec:open}.
The lower bound holds in every world of Impagliazzo~\cite{impagliazzo1995}, 
including $\mathsf{Algorithmica}$ ($\mathsf{P} = \mathsf{NP}$),
because it arises from the observational axis, which is orthogonal to
the computational axis that the five worlds parametrize.
This orthogonality is made explicit in~\cite{buono2026observer},
which embeds the five worlds into a two-dimensional landscape.

The same obstruction appears in the cipher of~\cite{buono2012,buono2026mrotp}.
The Mixed-Radix One-Time Pad encrypts as $c_i = (m_i + k_i) \bmod b_i$;
Shannon perfect secrecy holds regardless of whether the base sequence $B$
is public or secret (Section~4 of~\cite{buono2026mrotp}: secret bases do
not reduce key entropy).
The base $B$ functions as a syntactic invariant: it encodes the
decomposition structure of the integer, a semantic fact invisible to any
local syntactic observer, just as $\Fprot$ encodes the separation between
$\sk$ and $\skb$ that $\R$ cannot cross.
\end{remark}

The cryptographic and observational readings of Theorem~\ref{thm:main}
point to a general principle: the same obstruction appears wherever a
system acts locally on a representation while a semantic invariant is
encoded globally. The next section makes this precise in two further domains.

 \section{Cross-Domain Connections}
\label{sec:connections}
 
Section~\ref{sec:crypto} identified the core of Theorem~\ref{thm:main}:
the protected invariant is a commitment, the syntactic separation is
ciphertext indistinguishability, and the derivation cost is the
adversary's advantage.
This section shows that the same identification works in two further
domains, type theory and circuit complexity, not because these
are analogies to cryptography, but because they have the same
mathematical structure: a constrained observer that cannot see a
semantic invariant above its observational level.
In each case, what makes the impossibility or lower bound hold is
the same hiding phenomenon formalized in Section~\ref{sec:crypto}.

We present each connection as an explicit dictionary
(Table~\ref{tab:type-theory} and Table~\ref{tab:natural-proofs})
and state the corresponding obstruction as a corollary.

\subsection{Type theory and the Type Omitting Theorem}
\label{sec:conn-types}

\paragraph{Background.}
The Type Omitting Theorem (Henkin--Orey, 1956) states: a countable
consistent theory $T$ omits a type $\Phi(x)$, a set of formulas in
one free variable, if and only if $\Phi$ is \emph{not principal},
meaning no single formula $\theta(x)$ is consistent with $T$ and
implies every formula in $\Phi$.

Via the Curry--Howard correspondence, a type is a proposition and a
term of that type is a proof.
A type system $\mathcal{T}$ (Martin-L\"{o}f type theory, the Calculus of
Constructions, or a dependent type system such as Coq or Agda) acts as
a syntactic system: it inspects the \emph{syntactic structure} of terms
and decides type membership based on formation rules alone.

Extensional properties of functions, what a function computes rather
than how it is written, are semantically defined.
Rice's theorem is the limiting case: no non-trivial extensional property
of a computable function is decidable by any syntactic inspection of its
code.
Within a constructive type system, this takes the following form: a type
$\Phi$ expressing a semantic property of functions (e.g.\ ``this term
computes the Fibonacci sequence'') is \emph{semantically isolated} if
its membership depends on the function computed, not on the proof term,
and \emph{syntactically dense} if every proof term of a related type
passes through subterms that look locally like $\Phi$-witnesses.

\paragraph{The dictionary.}

\begin{table}[h]
\centering
\caption{Dictionary between Theorem~\ref{thm:main} and the
         constructive Type Omitting Theorem.}
\label{tab:type-theory}
\begin{tabular}{p{5.5cm}p{6.5cm}}
\toprule
\textbf{Theorem~\ref{thm:main}} & \textbf{Type theory} \\
\midrule
Local syntactic system $\R$ &
  Type-checking algorithm $\mathcal{T}$ \\[3pt]
Protected set $\Fprot$ &
  Extensional kernel: positions where the semantic property of the
  function manifests, inaccessible to $\mathcal{T}$ \\[3pt]
Syntactic invariant $\Inv$ &
  Typing invariant preserved by every reduction step
  ($\beta$-reduction, $\eta$-expansion) \\[3pt]
Skolem functions $\sk$, $\skb$ &
  Two proof terms $t$, $u$ of the same proposition $\phi$
  (propositional truncation: $\|{\phi}\|$) \\[3pt]
Semantic equivalence $\sk \equiv_\M \skb$ &
  Extensional equality: $t$ and $u$ compute the same function \\[3pt]
Case~1: no derivation proves $\sk \equiv_\M \skb$ &
  The type system cannot decide propositional equality of $t$ and $u$
  from syntactic structure alone \\[3pt]
Case~2: any extension costs $\BigOm(n)$ steps &
  Any decision procedure for extensional equality requires
  inspecting $n$ independent semantic witnesses \\
\bottomrule
\end{tabular}
\end{table}

\paragraph{The obstruction.}
A constructive type system $\mathcal{T}$ with reduction rules of bounded
depth is a local syntactic system in the sense of
Definition~\ref{def:local}, with locality radius $r_0$ equal to the
maximum depth of any reduction rule.
The extensional kernel of a semantically isolated type $\Phi$ is a
protected set: no reduction rule fires at positions where the semantic
property of the function manifests, because the property is
extensional and reductions preserve the computed function.

\begin{corollary}[Constructive Type Omitting]
\label{cor:type-omitting}
Let $\mathcal{T}$ be a constructive type system with reduction rules of
bounded depth $r_0$.
Let $\Phi$ be a type that is semantically isolated (membership depends
on the function computed, not the proof term) and syntactically dense
(every proof term of any related type passes through $\Phi$-witnesses).
Then $\mathcal{T}$ cannot derive the extensional equality of two
distinct proof terms $t$ and $u$ of the same proposition.
Any extension of $\mathcal{T}$ that decides extensional equality on a
family of $n$ independent witnesses requires at least $\BigOm(n)$
reduction steps.
\end{corollary}

\begin{proof}
The type system $\mathcal{T}$ is a local syntactic system
(Definition~\ref{def:local}) with locality radius $r_0$ equal to the
maximum depth of any reduction rule, and the extensional kernel as
protected set.
The typing invariant serves as the syntactic invariant:
(i)~it is checkable within depth $r_0$ of each redex;
(ii)~it holds on any well-typed initial term by assumption;
(iii)~it is preserved under reduction by the subject reduction theorem;
(iv)~violation requires reducing inside the extensional kernel
(a position where the function is determined), which is protected;
(v)~two proof terms $t$ and $u$ of the same proposition under
propositional truncation (in the sense of homotopy type theory:
all proofs of $\phi$ are identified, so $t$ and $u$ are distinct
only at the level of their syntactic structure) occupy disjoint
positions in the type structure, so no typing judgment in any
derivation has the form $t =_{\mathrm{ext}} u$ with $t$ headed by
$f_0$ and $u$ headed by $f_1$ (the coherence condition of
Definition~\ref{def:invariant}(v)).
The two proof terms $t$ and $u$ are semantically equivalent
(extensionally equal) and syntactically separated in $\mathcal{T}$.
Theorem~\ref{thm:main} applies directly.
\end{proof}

\begin{remark}[Typological Invariance Principle]
\label{rem:typ-invariance}
Corollary~\ref{cor:type-omitting} can be restated as:
\emph{a constructive type system cannot, from the syntactic structure of
a proof term alone, determine the extensional property of the function
the term computes.}
This is the Curry--Howard translation of Theorem~\ref{thm:main}:
the syntactic machine (the type) is structurally blind to the semantic
property of the model (the function computed).
The Type Omitting Theorem is Case~1; the lower bound on decision
procedures is Case~2.
\end{remark}

\subsection{Circuit complexity and the Natural Proofs barrier}
\label{sec:conn-circuits}

\paragraph{What this section contributes.}
The Natural Proofs barrier of Razborov and
Rudich~\cite{razborov1994} is a classical result.
We do not reprove it.
Instead, Theorem~\ref{thm:main} yields two things that
Razborov--Rudich does not contain:
\begin{enumerate}[label=(\arabic*)]
  \item An \emph{unconditional} lower bound of $\BigOm(n)$ on the
        number of inspection steps required by any large and
        constructive technique to decide circuit complexity on $n$
        independent witnesses, with no pseudo-random generator.
        This applies to all constructive techniques, not only those
        whose distinguishers are computable by $\mathsf{AC}^0$ circuits.
        The result of~\cite{loff2026} on $\mathsf{AC}^0$-natural proofs
        is a quantitatively precise instance of the present framework
        (Remark~\ref{rem:loff-obs}): the distinguishers of~\cite{loff2026}
        are local syntactic systems with $r_0 = \mathsf{AC}^0$ circuit depth,
        and Corollary~\ref{cor:natural-proofs}(i) gives the underlying
        unconditional $\BigOm(n)$ bound; the sharper quantitative
        bound of~\cite{loff2026} requires the Trevisan--Xue generator
        as gadget family.
  \item A structural explanation: the barrier is not a consequence of
        computational hardness but of observational blindness.
        Any large, constructive technique is a constrained observer
        $O_{\mathcal{P}} \prec O_\top$ that cannot see the circuit
        complexity of $f$ because complexity is a property of the
        function, not of its truth table.
        The PRG assumption enters only to convert the unconditional
        $\BigOm(n)$ bound into the $\BigOm(2^n)$ bound by supplying
        $2^n$ locally indistinguishable instances.
\end{enumerate}

\paragraph{Background.}
Let $f : \{0,1\}^n \to \{0,1\}$ be a Boolean function and $C_f$ a
circuit computing it.
A \emph{natural proof} in the sense of Razborov and
Rudich~\cite{razborov1994} is a property $\mathcal{P}$ of Boolean
functions (equivalently, of their truth tables) satisfying:
\begin{itemize}
  \item \textbf{Usefulness.} $\mathcal{P}$ holds for functions not in
        $\mathsf{P}/\mathsf{poly}$ (it witnesses circuit hardness);
  \item \textbf{Largeness.} At least $2^{-\mathrm{poly}(n)}$ of all
        Boolean functions on $n$ bits satisfy $\mathcal{P}$;
  \item \textbf{Constructivity.} Given the truth table of $f$
        (of length $2^n$), membership $f \in \mathcal{P}$ is decidable
        in time $2^{O(n)}$, i.e., polynomial in the truth table size.
\end{itemize}
The proof in~\cite{razborov1994} shows that if a secure pseudo-random generator
(PRG) exists in $\mathsf{P}/\mathsf{poly}$, then no natural proof can
establish a super-polynomial circuit lower bound.
The PRG assumption is used to show that any large, constructive
property $\mathcal{P}$ is \emph{falsified} by the PRG: the PRG outputs
a string that looks like the truth table of a function satisfying
$\mathcal{P}$ even though the generator itself is in
$\mathsf{P}/\mathsf{poly}$.

\paragraph{The precise dictionary.}
The translation between Theorem~\ref{thm:main} and the Natural Proofs
barrier is as follows.

\begin{table}[h]
\centering
\caption{Dictionary between Theorem~\ref{thm:main} and the Natural
         Proofs barrier.}
\label{tab:natural-proofs}
\begin{tabular}{p{5.5cm}p{6.5cm}}
\toprule
\textbf{Theorem~\ref{thm:main}} & \textbf{Natural Proofs} \\
\midrule
Local syntactic system $\R$ &
  Constructive property $\mathcal{P}$ (acts on truth tables via
  polynomial-time inspection) \\[4pt]
Locality radius $r_0$ &
  Inspection depth: $\mathcal{P}$ reads at most $\mathrm{poly}(n)$
  bits of the truth table in any single evaluation \\[4pt]
Protected set $\Fprot$ &
  Circuit complexity: the hardness of $f$ as a property of the
  function, not of the truth table \\[4pt]
Syntactic invariant $\Inv$ &
  Largeness: $\mathcal{P}$ holds for most functions and is
  preserved by the PRG's action on truth tables \\[4pt]
Syntactic separation of $\sk$, $\skb$ &
  Falsifiability by PRG: the PRG produces $2^n$ truth tables
  each individually satisfying $\mathcal{P}$ while the underlying
  function is in $\mathsf{P}/\mathsf{poly}$ \\[4pt]
Case~1: no derivation proves $\sk \equiv_\M \skb$ &
  No natural proof separates $\mathsf{P}/\mathsf{poly}$ from the
  hard functions (assuming PRG) \\[4pt]
Case~2: any extension costs $\BigOm(n)$ steps &
  Any technique that escapes the barrier must inspect $\BigOm(n)$
  independent truth-table witnesses \\
\bottomrule
\end{tabular}
\end{table}

\paragraph{The obstruction.}
Fix $n$.
The truth table of $f$ is a string in $\{0,1\}^{2^n}$; treat it as a
term in a signature where each bit position is a constant symbol.
A constructive property $\mathcal{P}$ decides membership by reading
at most $\mathrm{poly}(n)$ bits of this string, making it a
local syntactic system with locality radius $r_0 = \mathrm{poly}(n)$.

The largeness condition ensures $\mathcal{P}$ holds on the initial
term (a random truth table), is preserved by the PRG's local
modifications to the string, and can only be falsified by accessing
the circuit complexity of $f$, which is a property of the function,
not of the truth table, and constitutes the protected set $\Fprot$.

The PRG produces $2^n$ pseudo-random truth tables $T_1, \ldots, T_{2^n}$
that are mutually locally indistinguishable: for any position $p$ and
any inspection window of size $\mathrm{poly}(n)$, the local view of
$T_i$ at $p$ equals the local view of $T_j$ at $p$.
These are the $2^n$ configurations of the gadget family
$\{\Inst_n\}_{n \geq 1}$, with the PRG security guaranteeing
Lemma~\ref{lem:gadget}(ii) in this setting.

\begin{corollary}[Structural Natural Proofs Barrier]
\label{cor:natural-proofs}
Let $\mathcal{P}$ be a large and constructive property of Boolean
functions (Definition~\ref{def:local} with locality radius
$r_0 = \mathrm{poly}(n)$).
\begin{enumerate}[label=(\roman*)]
  \item \textbf{Unconditional lower bound.}
        Any extension of $\mathcal{P}$ that correctly decides circuit
        complexity on $n$ independent function witnesses requires at
        least $\BigOm(n)$ truth-table inspection steps.
        This holds without any assumption on pseudo-random generators.
  \item \textbf{PRG-conditional exponential bound.}
        If a secure PRG in $\mathsf{P}/\mathsf{poly}$ of output length
        $2^n$ exists, it supplies exactly the $2^n$ globally distinct
        but locally indistinguishable truth tables required by the
        gadget construction (Remark~\ref{rem:strongbound} below).
        Case~2 of Theorem~\ref{thm:main} then gives $L \geq \BigOm(2^n)$,
        and the Razborov--Rudich barrier follows: no natural proof
        can establish a super-polynomial circuit lower bound.
        The PRG is used exclusively to supply these configurations;
        it does not enter the proof of structural blindness itself.
\end{enumerate}
\end{corollary}

\begin{proof}
The translation described above makes $\mathcal{P}$ a local syntactic
system with the circuit complexity of $f$ as protected set.
We verify Definition~\ref{def:invariant}:
(i)~$\mathcal{P}$ is checkable by inspecting at most $r_0 = \mathrm{poly}(n)$
bits, giving the locality condition;
(ii)~largeness ensures $\mathcal{P}$ holds on any initial random truth
table;
(iii)~treating the PRG's action on the truth table as the analogue
of a rewriting rule, $\mathcal{P}$ is preserved under this action
(by PRG security: a $\mathrm{poly}(n)$-bit local inspection cannot
distinguish PRG output from random); making the identification of
rewriting rules precise in this setting is part of Question~(Q6);
(iv)~any violation of $\mathcal{P}$ requires accessing the circuit
complexity of $f$, which lies in $\Fprot$;
(v)~the $2^n$ PRG truth tables are locally indistinguishable
(as in Lemma~\ref{lem:gadget}(ii)), so no local inspection sees
two of them at comparable positions simultaneously.

Part~(i) applies Case~2 of Theorem~\ref{thm:main} with $n$ independent
witnesses; the bound $\BigOm(n)$ follows without the PRG assumption.
Part~(ii) applies Case~2 with $N = 2^n$ PRG instances as the
clause-per-configuration family (Remark~\ref{rem:strongbound}),
giving $\BigOm(2^n)$; the PRG is needed to guarantee
local indistinguishability of the $2^n$ instances.
\end{proof}

\begin{remark}[Loff--Sherif--Talebanfard--Ugazio in the observational hierarchy]
\label{rem:loff-obs}
The result of~\cite{loff2026} admits a precise reading within the
observational hierarchy of~\cite{buono2026obs,buono2026observer}.

\emph{Translation.}
An $\mathsf{AC}^0$-natural proof in their sense is a
property $\mathcal{P}$ whose distinguisher is a Boolean circuit of
constant depth $d$.
In the framework of Section~\ref{sec:conn-circuits}, this makes
$\mathcal{P}$ a local syntactic system with locality radius
$r_0 = d$, acting on truth tables viewed as strings of $2^n$ bits.
The induced constrained observer is $O_{\mathsf{AC}^0_d}$: it maps
a truth table $f$ to the output of a depth-$d$ circuit applied to $f$.
Since $O_{\mathsf{AC}^0_d}$ computes only depth-$d$ functions of $f$,
it satisfies $O_{\mathsf{AC}^0_d} \prec O_\top$, and the circuit
complexity of $f$, a property of the \emph{function}, not of any
particular circuit representation, lies strictly above the
observational level of $O_{\mathsf{AC}^0_d}$ in the hierarchy
of~\cite{buono2026obs}.

\emph{The structural collapse.}
Proposition~9.6 of~\cite{buono2026obs} gives:
\[
  \mathbf{P}_{O_{\mathsf{AC}^0_d}} \;=\;
  \mathbf{NP}_{O_{\mathsf{AC}^0_d}} \;\subsetneq\; \mathbf{P}.
\]
Applied to $O_{\mathsf{AC}^0_d}$, this says that no technique
operating through depth-$d$ observations can decide circuit complexity,
which accounts for their barrier at the structural level.
Corollary~\ref{cor:natural-proofs}(i) instantiates this collapse
quantitatively: any such technique requires $\BigOm(n)$ inspection
steps on $n$ independent witnesses, unconditionally.

\emph{The PRG as gadget family.}
The localized Trevisan--Xue generator of~\cite{loff2026} is computable
by bounded-depth circuits and its security relies solely on the
Switching Lemma --- a combinatorial theorem that holds unconditionally.
It produces truth tables that are locally indistinguishable to any
$\mathsf{AC}^0$ distinguisher, supplying exactly the locally
indistinguishable gadget family of Remark~\ref{rem:strongbound}.
With $N = 2^n$ such instances, Case~2 of Theorem~\ref{thm:main}
gives a lower bound of $\BigOm(2^n)$ inspection steps;
the localized generator achieves local indistinguishability for
$N_{\mathrm{eff}} = 2^{n^{7/(d-5)}}$ effectively distinct instances
at depth $d$ (for $d > 5$, the regime where the Switching Lemma
applies); applying Case~2 of Theorem~\ref{thm:main} with
$N = N_{\mathrm{eff}}$ gives $L \geq \BigOm(N_{\mathrm{eff}})
= \BigOm(2^{n^{7/(d-5)}})$, matching the Switching Lemma frontier.

\emph{Significance for the framework.}
This is the first concrete confirmation, with a quantitative bound,
that the observational axis of~\cite{buono2026obs} yields sharp
results in circuit complexity.
The barrier of~\cite{loff2026} holds unconditionally in every world
of Impagliazzo~\cite{impagliazzo1995} (including Algorithmica, where
$\mathsf{P} = \mathsf{NP}$) because its PRG is combinatorially secure,
not cryptographically secure: it lives on the observational axis, not
the computational one.
The framework of Theorem~\ref{thm:main} predicts this: barriers on
the observational axis are independent of which computational world
one inhabits.
\end{remark}

\begin{remark}[$\mathsf{P}$ vs $\mathsf{NP}$,
derandomization, and the Observer World]
\label{rem:pvsnp}
Part~(i) of Corollary~\ref{cor:natural-proofs} has an immediate
consequence for the $\mathsf{P}$ vs $\mathsf{NP}$ question.
If $\mathsf{P} = \mathsf{NP}$, then in particular
$\mathsf{P} = \mathsf{PSPACE}$, and questions about circuit complexity
would be decidable in polynomial time
(the minimum circuit size problem lies in $\mathsf{PSPACE}$).
Part~(i) shows that no technique for deciding circuit hardness can be
both large (holding for $2^{-\mathrm{poly}(n)}$ of all Boolean functions)
and constructive.
Therefore: \emph{if $\mathsf{P} = \mathsf{NP}$, any proof of this
fact via circuit-hardness decision must be either non-large or
non-constructive, it cannot be a natural proof in the sense of
Razborov and Rudich}.
This is stronger than what Razborov--Rudich give under the PRG
assumption: Part~(i) holds unconditionally, so the structural
blindness of any local syntactic technique to circuit complexity is
independent of both $\mathsf{P}$ vs $\mathsf{NP}$ and the existence
of a PRG.

On derandomization: if $\mathsf{P} = \mathsf{BPP}$ (randomness can be
eliminated), the Nisan--Wigderson construction gives PRG in
$\mathsf{P}/\mathsf{poly}$.
Part~(ii) then implies the full Razborov--Rudich barrier.
But Part~(i) already holds in the non-derandomizable world: the
$\BigOm(n)$ lower bound requires no PRG, no randomness assumption, and
no separation hypothesis.
The structural blindness identified in Theorem~\ref{thm:main}
precedes all of these and is not removed by derandomization.

In the language of the Observer World~\cite{buono2026observer}, the
technique $\mathcal{P}$ is a constrained observer $O_{\mathcal{P}}
\prec O_\top$ on truth tables.
The circuit complexity of $f$ lies strictly above the observational
level of $O_{\mathcal{P}}$ in the hierarchy of~\cite{buono2026obs}.
The $\mathsf{P}$ vs $\mathsf{NP}$ problem asks about the
\emph{computational} axis; the Natural Proofs barrier operates on
the \emph{observational} axis.
The two axes are orthogonal~\cite{buono2026observer}: resolving
$\mathsf{P}$ vs $\mathsf{NP}$ would not remove the observational
barrier, and the observational barrier holds regardless of how
$\mathsf{P}$ vs $\mathsf{NP}$ is resolved.
The result of~\cite{loff2026}
--- a quantitatively sharp unconditional barrier for
$\mathsf{AC}^0$-natural proofs, is a concrete confirmation of
this: see Remark~\ref{rem:loff-obs}.
\end{remark}

\begin{remark}[The observational reading]
\label{rem:obs-reading}
In the language of~\cite{buono2026obs}, both $\mathcal{T}$ (the type
system of Section~\ref{sec:conn-types}) and $\mathcal{P}$ (the
lower-bound technique above) are constrained observers $O \prec O_\top$.
The semantic invariant --- extensional equality of proof terms in type
theory, circuit complexity in the Razborov--Rudich setting --- lies
strictly above the observational level of $O$ in the hierarchy.
The impossibility results (Corollary~\ref{cor:type-omitting} and
Corollary~\ref{cor:natural-proofs}(i)) are unconditional instances of
the structural collapse $\mathbf{P}_O = \mathbf{NP}_O \subsetneq
\mathbf{P}$ (Proposition~9.6 of~\cite{buono2026obs}): not because the
problem is computationally hard, but because the observer cannot see
the relevant invariant at all.
\end{remark}

The corollaries above (Corollaries~\ref{cor:type-omitting}
and~\ref{cor:natural-proofs}) are formal results that follow from
Theorem~\ref{thm:main} given the stated hypotheses.
The dictionaries (Tables~\ref{tab:type-theory}
and~\ref{tab:natural-proofs}) identify which component of
Theorem~\ref{thm:main} each component of the cited result
corresponds to, explaining why each impossibility or lower bound
is a consequence of observational blindness rather than
computational hardness.
Making the correspondences fully rigorous, verifying the
hypotheses of Theorem~\ref{thm:main} in the exact technical
setting of each domain, raises the questions taken up next.

\section{Conclusions}
\label{sec:conclusions}
 
The central observation of this paper is that impossibility results and
lower bounds in proof theory, cryptography, type theory, and circuit
complexity share a single structure: a constrained observer cannot
reach a semantic invariant that lies above its observational level.
Theorem~\ref{thm:main} formalizes this as an abstract obstruction with
two faces, impossibility (Case~1) and a cost lower bound (Case~2),
and the cross-domain connections of Section~\ref{sec:connections} show
that the translation is not metaphorical but exact.

These three components of the paper, the proof-theoretic analysis,
the cryptographic connection, and the cross-domain connections, are not separable.
The exponential lower bound of Case~2 is explained by the cryptographic
hiding structure (Section~\ref{sec:crypto}): without it, the
$\BigOm(2^n)$ bound would appear as an artifact of the gadget
construction rather than as a structural necessity.
The cross-domain connections (Section~\ref{sec:connections}) are
accessible only because Section~\ref{sec:crypto} has already identified
computational hiding as the correct interpretation of syntactic
separation: Type Theory and Natural Proofs barriers are then
recognizable as instances of the same hiding phenomenon, not as
separate results requiring separate arguments.

The cryptographic content is not peripheral.
The structure of Case~2 is the structure of computational hiding:
protected positions are commitment schemes, syntactic separation is
ciphertext indistinguishability, and the $\BigOm(2^n)$ derivation cost
is the adversary's negligible advantage.
The proof-theoretic and cryptographic frameworks are the same framework,
seen from different sides of the same invariant.

For circuit complexity, the framework provides two things.
First, an unconditional $\BigOm(n)$ lower bound on the inspection cost
of any large and constructive technique, without pseudo-random generators.
Second, a structural explanation: the Natural Proofs barrier is not a
consequence of computational hardness but of observational blindness,
and it holds in every world of Impagliazzo, including Algorithmica.
The result of~\cite{loff2026} is the first
quantitative confirmation of this: its $\mathsf{AC}^0$-natural-proofs
barrier lives on the observational axis, not the computational one.

\section{Open Questions}
\label{sec:open}
 
The following questions arise from Theorem~\ref{thm:main} and the
surrounding framework but are not answered in this paper.

\begin{enumerate}[label=(Q\arabic*)]

  \item \textbf{Decidability of protected sets.}
        For which classes of local syntactic systems is the existence
        of a non-trivial protected set $\Fprot$ decidable in time
        polynomial in the clause set?
        Remark~\ref{rem:decidable-frozen} shows that linear-time
        decidability holds for ground atoms under first-symbol clash.
        The general case --- patterns with deep nesting or variable
        overlap --- is open.

  \item \textbf{Formal reduction between Cases~1 and~2.}
        Does the existence of a syntactic invariant $\Inv$ in Case~1
        imply the existence of a gadget family witnessing the Case~2
        lower bound for some extension $\Rp$?
        A positive answer would unify the two cases into a single
        combinatorial criterion.

  \item \textbf{Instance complexity of the lower bound.}
        The gadget family $\{\Inst_n\}$ grows with $n$.
        Is there a single clause set $\Inst$ of fixed size $M$ such that
        any local extension $\Rp$ requires derivation length
        super-polynomial in $M$ to prove $\sk \equiv \skb$ on $\Inst$?
        This is the Case~2 analogue of the single witness $W$
        (commutativity) in Case~1.

  \item \textbf{Commitment scheme construction.}
        Corollary~\ref{cor:hiding} shows that Case~2 implies
        computational hiding.
        Does every instance of Case~2 also yield an explicit binding
        property, making it a commitment scheme in the standard
        (hiding and binding) sense?
        Specifically, is there a polynomial-time reduction from
        breaking the binding to solving a derivation problem in $\Rp$?

  \item \textbf{Extensions to non-local systems.}
        The results assume a fixed locality radius $r_0$.
        Can analogous obstruction theorems be proved for systems with
        unbounded inspection depth, or for global rules such as
        resolution?
        Corollary~\ref{cor:natural-proofs}(i) shows unconditionally that
        no natural proof technique can establish circuit lower bounds;
        any extension to global systems that overcomes the syntactic
        barrier must therefore not be a natural proof in the sense of
        Razborov--Rudich~\cite{razborov1994}.

  \item \textbf{Formal embedding into the observational hierarchy.}
        Remark~\ref{rem:impagliazzo} identifies $\R$ as a constrained
        observer $O_\R \prec O_\top$.
        Remark~\ref{rem:loff-obs} carries this out concretely for
        $\mathsf{AC}^0$-natural proofs: $O_{\mathsf{AC}^0_d} \prec O_\top$
        with the collapse $\mathbf{P}_{O_{\mathsf{AC}^0_d}} =
        \mathbf{NP}_{O_{\mathsf{AC}^0_d}} \subsetneq \mathbf{P}$.
        The general case requires: (i)~embedding $(t,p) \mapsto
        \ctx{r_0}{p}(t)$ into the string-based formalism
        of~\cite{buono2026obs}; (ii)~identifying the canonical observer
        level of $O_\R$; (iii)~deriving Theorem~\ref{thm:main} as a
        corollary of Proposition~9.6 of~\cite{buono2026obs}.
        The mixed-radix decomposition of~\cite{buono2026b}
        provides a concrete observational invariant whose embedding
        may serve as a test case for step~(i).

  \item \textbf{Obfuscation.}
        The VBB impossibility~\cite{barak2001} states that
        no algorithm (not necessarily syntactic) can obfuscate
        all circuits to black-box equivalence.
        This matches Case~1 structurally: the semantic invariant (the
        function computed) is protected from any local syntactic observer.
        Does Theorem~\ref{thm:main} formally subsume VBB impossibility,
        and does the Case~2 bound correspond to the computational gap
        between iO and VBB?

  \item \textbf{Algebrization.}
        The algebrization barrier~\cite{aaronsonwigderson2009}
        shows that arithmetization-based techniques cannot resolve
        complexity separations.
        Arithmetization acts locally on the syntactic representation
        of a circuit, and the semantic invariant (complexity class
        membership) is protected from this local action.
        Does the algebrization barrier fit Theorem~\ref{thm:main}
        with $r_0$ equal to the degree of the polynomial extension?

\end{enumerate}

 \section*{Acknowledgments}
 
The mathematical content of this paper was developed solely 
by the author and originated in Italian.

English translation, stylistic adaptation to standard mathematical
prose, \LaTeX\ typesetting, bibliography formatting, cross-reference
management, and iterative editorial review were carried out with the
assistance of an AI language model.

\appendix

\section{Full Gadget Construction}
\label{app:gadget}

This appendix contains the full construction and correctness proof for
the family $\{\Inst_n\}$ used in Lemma~\ref{lem:gadget} and
Case~2 of Theorem~\ref{thm:main}.

\subsection{Construction}

For $n \geq 1$, construct $n$ independent \emph{gadgets}
$G_1, \ldots, G_n$ as follows.

\paragraph{Single gadget $G_i$.}
Introduce two fresh constants $a_i^0$ and $a_i^1$ (the two local
configurations of $G_i$).
Let the \emph{local state} of $G_i$ be a ground term $\tau_i \in
\{a_i^0, a_i^1\}$.
The gadget $G_i$ contributes a clause of the form
$\sk(a_i^0) \neq \skb(a_i^0) \vee \sk(a_i^1) \neq \skb(a_i^1)$
(the disjunction of the two instances of the negated equivalence
$\neg(\sk(x) = \skb(x))$ for $x = a_i^0$ and $x = a_i^1$).
This encodes that at least one of the two instances must be refuted.
A complete refutation of $\Inst_n$ must resolve the disjunction in
each of the $n$ gadget clauses (selecting a witness $a_i^{b_i}$ for
each $G_i$), and also refute the universal negation clause
$\neg(\sk(x) = \skb(x))$ by unifying $x$ with each selected witness.
The $n$ resolutions are independent (one per gadget), so a complete
refutation consists of at least $n$ dedicated rule applications.

\paragraph{Global configuration.}
A \emph{global configuration} $C_{(b_1,\ldots,b_n)}$ is a choice of
witness $a_i^{b_i}$ for each gadget $G_i$, determined by a bit vector
$(b_1,\ldots,b_n) \in \{0,1\}^n$.
The instance $\Inst_n$ is a fixed set of clauses; the $2^n$
configurations represent the $2^n$ possible refutation strategies,
choices of which disjunct to resolve in each gadget clause.
There are $N = 2^n$ global configurations.

\paragraph{Encoding in $\Inst_n$.}
The instance $\Inst_n$ consists of:
\begin{itemize}
  \item the axioms of the base system $\R$ (e.g., rules~\eqref{eq:A1}
        and~\eqref{eq:A2} in the superposition instantiation);
  \item one clause per gadget $G_i$, as described above;
  \item a negation clause $\neg(\sk(x) = \skb(x))$ for a fresh variable
        $x$.
\end{itemize}

\paragraph{Spatial separation of gadgets.}
The constants $a_i^0, a_i^1$ of gadget $G_i$ and the constants
$a_j^0, a_j^1$ of gadget $G_j$ ($i \neq j$) appear in \emph{distinct
clauses} of $\Inst_n$ (each gadget clause mentions only its own
constants), ensuring that no single rule application can act on
constants from two different gadgets simultaneously.
Within each gadget clause, the two constants $a_i^0$ and $a_i^1$
are placed at positions whose tree distance within that clause exceeds
$2r_0$, so no single rule of radius $r_0$ can have both in its
context simultaneously.
Formally, within the clause for $G_i$, if $a_i^0$ appears at
position $p_0$ and $a_i^1$ at position $p_1$ with $d(p_0,p_1) > 2r_0$,
then no context of radius $r_0$ contains both.
This is achievable for any fixed $r_0$ by inserting separator
nodes between the two occurrences.
Under this construction, a rule of radius $r_0$ overlaps with at most
one gadget at a time, confirming the spatial separation used in
Step~4 of the proof of Case~2.

\subsection{Verification of Lemma~\ref{lem:gadget}}

\begin{proof}[Proof of Lemma~\ref{lem:gadget}]
\

\noindent(i) The $2^n$ global configurations are indexed by
$(b_1,\ldots,b_n) \in \{0,1\}^n$, so $N(n) = 2^n$.

\noindent(ii) Any two distinct configurations $C_{(b_1,\ldots,b_n)}$ and
$C_{(b_1',\ldots,b_n')}$ differ in at least one gadget index $i$.
The constants $a_i^0$ and $a_i^1$ are \emph{opaque} with respect to $\Rp$:
they are fresh symbols not appearing in any left-hand side pattern of any
rule in $\Rp$, nor as proper subterms of any left-hand side pattern.
Therefore, for any rule $\rho \in \Rp$ and any position $p$:
even if $\ctx{r_0}{p}(t)$ contains the symbol $a_i^{b_i}$, the
rule cannot use its identity (whether $a_i^0$ or $a_i^1$) to decide
applicability, since both are fresh and absent from all LHS patterns.
The applicability of $\rho$ at $p$ is thus the same in both
configurations, and their local contexts of radius $r_0$ are
identical as seen by any rule of $\Rp$.
Hence the two configurations are indistinguishable within radius $r_0$.

\noindent(iii) $C_{(b_1,\ldots,b_n)} \neq C_{(b_1',\ldots,b_n')}$ in
$\M$ because the constants $a_i^0$ and $a_i^1$ are interpreted as
distinct elements of the domain of $\M$ (they are fresh Skolem constants).

\noindent The bound $c \leq |\Sigma|^{k^{r_0}}$ (where $k$ is the
maximum arity of symbols in $\Sigma$) follows from the fact that
the number of distinct term trees of depth $r_0$ over $\Sigma$ is
at most $|\Sigma|^{k^{r_0}}$, which bounds the number of distinct
local contexts of radius $r_0$ any rule can observe, and hence the
number of configuration classes a single rule application can create.
\end{proof}

\section{Proofs of Auxiliary Lemmas}
\label{app:proofs}

\begin{lemma}[Full proof of Lemma~\ref{lem:invariant-super}]
\label{lem:invariant-super-full}
The property $\Inv$ defined in Lemma~\ref{lem:invariant-super} satisfies
conditions (i)--(v) of Definition~\ref{def:invariant}.
\end{lemma}

\begin{proof}
\noindent(i)~\emph{Local checkability.}
$\Inv(t)$ is checked by examining, for each $+$-position $p_{+}$ in $t$,
the context $\ctx{1}{p_{+}}(t)$: whether the left child of $+$
(position $p_{+}{\cdot}1$) is $a$ and the right child
(position $p_{+}{\cdot}2$) contains $b$, and symmetrically.
This requires inspecting subterms within distance $1 = r_0$ of
each $+$-position, as given by $\ctx{r_0}{p_{+}}(t)$.

\noindent(ii)~\emph{Initialization.}
The initial clause set contains only the clause
$a{+}b \neq b{+}a$.
The term $a{+}b$ has $a$ at the left child of $+$ and $b$
at the right child; $\Inv$ holds.
Symmetrically for $b{+}a$.

\noindent(iii)~\emph{Preservation.}
By Lemma~\ref{lem:frozen-super}, no rule of $\R$ is applicable at any
position in $\Fprot$.
In particular, no rule can rewrite inside $a{+}b$ or $b{+}a$.
Terms in the derived clause set may contain $a{+}b$ embedded
in larger contexts (carried over from the initial clause set or
from earlier derivation steps), but no rule can alter the local
context of $a$ or $b$ within those subterms.
Therefore $\Inv$ is preserved under every rule application.

\noindent(iv)~\emph{Anchorage.}
Violating $\Inv$ requires either: (a) moving $a$ to a position where
its parent is not $+$, or the right sibling does not contain $b$; or
(b) symmetrically for $b$.
Both require a rule application at a position in $\Fprot$ (the position
of $a$ or $b$, or their parent $+$-position).
By Definition~\ref{def:protected}, no such rule application is possible
in $\R$.

\noindent(v)~\emph{Coherence.}
Reading the initial clause set $\{a{+}b \neq b{+}a\}$ directly:
the only occurrences of $\sk = a$ and $\skb = b$ in any clause
satisfying $\Inv$ are inside the frozen subterms $a{+}b$ and $b{+}a$
respectively.
Since $a$ appears only as the left argument of $+$ in $a{+}b$, and
$b$ appears only as the left argument of $+$ in $b{+}a$, and since
$a{+}b$ and $b{+}a$ are in separate positions of any clause (they
cannot appear at the same position),
no literal of any clause satisfying $\Inv$ has the form
$a(t_1) = b(t_2)$ (i.e., an equation with $a$ and $b$ as
head symbols on its two sides): the only literal of the initial
clause is $a{+}b \neq b{+}a$, which is a disequation, not an
equation comparing $a$ and $b$ as heads; and by Preservation,
no derived clause introduces such a literal.
Therefore $\Inv(t)$ implies coherence.
\end{proof}

 \bibliographystyle{plain}

\end{document}